\documentclass{article}

\usepackage{microtype}
\usepackage{graphicx}
\usepackage{subfigure}
\usepackage{booktabs} 
\usepackage{amsmath, amssymb, amsfonts}
\usepackage[normalem]{ulem}
\usepackage{multirow}

\usepackage{hyperref}

\usepackage[accepted]{icml2020}

\newcommand{\codecomment}[1]{\textbf{\color{black}// #1}}

\icmltitlerunning{Confidence Sets and Hypothesis Tests in a Likelihood-Free Inference Setting}

\def\X{{\mathbf{X}}}
\def\x{{\mathbf{x}}}
\def\P{\mathbb{P}}
\def\I{\mathbb{I}}
\def\O{\mathbb{O}}
\def\Or{\mathbb{OR}}
\def\E{\mathbb{E}}

\def\L{{\mathcal L}}
\def\D{{\mathcal D}}
\def\T{{\mathcal T}}
\def\Pois{\textrm{Poisson}}

\usepackage{amsfonts, amsmath, amssymb, amsthm, bm, 
            latexsym, mathrsfs, thmtools, wasysym}
\usepackage{cleveref}
\declaretheorem[name=Theorem, refname={theorem, theorems}, Refname={Theorem, Theorems}, parent=section]{theorem}

\declaretheorem[name=Lemma, refname={lemma, lemmas}, Refname={Lemma, Lemmas}, sibling=theorem]{lemma}
\declaretheorem[name=Proposition, refname={prop, props}, Refname={Prop, Props}, sibling=theorem]{Prop}
\declaretheorem[name=Assumption, refname={assumption, assumptions}, Refname={Assumption, Assumption}, sibling=theorem]{assumption}

\begin{document}

\twocolumn[
\icmltitle{Confidence Sets and Hypothesis Testing in a Likelihood-Free Inference Setting}



\icmlsetsymbol{equal}{*}

\begin{icmlauthorlist}
\icmlauthor{Niccol\`o Dalmasso}{CMU}
\icmlauthor{Rafael Izbicki}{SaoCarlos}
\icmlauthor{Ann B.~Lee}{CMU}
\end{icmlauthorlist}

\icmlaffiliation{CMU}{Department of Statistics \& Data Science, Carnegie Mellon University, Pittsburgh, USA}
\icmlaffiliation{SaoCarlos}{Department of Statistics, Federal University of S\~ao Carlos, S\~ao Paulo, Brazil}
\icmlcorrespondingauthor{Niccol\`o Dalmasso}{ndalmass@stat.cmu.edu}

\icmlkeywords{Likelihood-free inference, Frequentist Statistics, Confidence Sets, Hypothesis Testing}

\vskip 0.3in
]



\printAffiliationsAndNotice{}  

\begin{abstract} 
Parameter estimation, statistical tests and confidence sets are the cornerstones of classical statistics that  allow  scientists to make inferences about the underlying process that generated the observed data. A key question is whether one can still construct hypothesis tests and confidence sets with proper coverage and high power in a   so-called likelihood-free inference (LFI) setting; that is, a setting where the likelihood is not explicitly known but one can forward-simulate observable data according to a stochastic model.
In this paper, we present \texttt{ACORE} (Approximate Computation via Odds Ratio Estimation), a frequentist approach to LFI that first formulates the classical likelihood ratio test (LRT) as a parametrized classification problem, and then uses the equivalence of tests and confidence sets to build confidence regions for parameters of interest. 
We also present a goodness-of-fit procedure 
for checking whether the constructed tests and confidence regions are valid. \texttt{ACORE} is based on the key observation that the LRT statistic, the rejection probability of the test, and the coverage of the confidence set are  conditional distribution functions which often vary smoothly as a function of the parameters of interest. Hence, instead of relying solely on  samples simulated at fixed parameter settings (as is the convention in standard Monte Carlo solutions), one can leverage machine learning tools and data simulated in the neighborhood of a parameter to improve estimates of quantities of interest. We demonstrate the efficacy of \texttt{ACORE} with both theoretical and empirical results. Our implementation is available on \href{https://github.com/Mr8ND/ACORE-LFI}{\texttt{Github}}.
 \end{abstract}

\section{Introduction}

Parameter estimation, statistical tests and confidence sets are the cornerstones of classical statistics that relate observed data to properties of the underlying statistical model. Most frequentist procedures with good statistical performance (e.g.,  high power) require explicit knowledge of a likelihood function. However, in many science and engineering applications, complex phenomena are modeled by forward simulators  that \textit{implicitly} define a likelihood function: For example, given input parameters $\theta$, a statistical model of our environment, climate or universe may combine deterministic dynamics with random fluctuations to produce synthetic data $\mathbf{X}$. Simulation-based inference without an explicit likelihood is called {\em likelihood-free inference} (LFI).

The literature on LFI is vast. 
Traditional LFI methods, such as Approximate Bayesian Computation (ABC; \citealt{beaumont2002approximate,marin2012approximate,sisson2018handbook}), estimate posteriors by using simulations sufficiently close to the observed data, hence bypassing the likelihood. More recently, several
approaches that  leverage machine learning algorithms have been proposed; these either directly estimate the posterior distribution \citep{marin2016abcrf,chen2019gaussiancopula,izbicki2019abc,greenberg2019posterior} or the likelihood function \citep{izbicki2014, Thomas2016likelihood,price2018bayesian, Ong2018lfilikelihood,lueckmann2019likelihood, papamakarios2019sequential}.
We refer the reader to \citet{Cranmer2019frontier_sim_based} for a recent review of the field. 

A question that has not received much attention so far is whether one, in an LFI setting, can construct inference techniques with good frequentist properties. 
Frequentist procedures have nevertheless played an important role in many fields. In high energy physics for instance,  classical statistical techniques (e.g., hypothesis testing for outlier detection)   have resulted in discoveries of new physics and other successful applications~\citep{Feldman1998UnifiedApproach, Cranmer2015PraticalStatsLHC, Cousins2018Lectures}.
Even though controlling type I error probabilities is important in these applications,
most LFI methods do not have guarantees on validity or power. 
Ideally, a unified LFI approach should 
\vspace{-3mm}
\begin{itemize}
\item be {\em computationally efficient} in terms of the number of required simulations,
\vspace{-2mm}
\item handle {\em high-dimensional data} from different sources (without, e.g., predefined summary statistics),
\vspace{-2mm}
\item  produce hypothesis tests and confidence sets that are  {\em valid}; that is, have the nominal type I error or confidence level,
\vspace{-2mm}
\item  produce hypothesis tests with {\em high power} or, equivalently, confidence sets with a small expected size,
\vspace{-2mm}
\item provide {\em diagnostics} for checking empirical coverage or for checking how well the estimated likelihood fits simulated data.
\end{itemize}
\vspace{-3mm}
In this paper, we present  \texttt{ACORE} (Approximate Computation via Odds Ratio Estimation), a frequentist approach to LFI, which addresses the above mentioned concerns.

Figure~\ref{fig:work_structure} summarizes the \texttt{ACORE} work structure: \texttt{ACORE} first compares synthetic data from the simulator $F_\theta$ to a reference distribution $G$ by computing an ``{Odds Ratio}''.
The odds ratio can be learnt with a probabilistic classifier, such as a neural network with a softmax layer, suitable for the data at hand. As we shall see, the estimated odds ratio is an approximation of the likelihood ratio statistic
(Proposition~\ref{prop::consistency}).
The \texttt{ACORE} test statistic (Equation~\ref{eq:odds_ratio_statistic}), together with an estimate of the ``{Critical Value}'' (Algorithms~\ref{alg:estimate_thresholds2} and \ref{alg:estimate_thresholds_conf2}), can be used for hypothesis testing or for finding a confidence set for $\theta$. \texttt{ACORE} also includes ``{Diagnostics}'' (Section~\ref{sec:coverage_gof}) for computing the empirical coverage of the constructed confidence set for $\theta$.

At the heart of \texttt{ACORE} is the key observation that the likelihood ratio statistic, the critical value of the test, and the coverage of the confidence set are  conditional distribution functions which often vary smoothly as a function of the (unknown) parameters of interest. Hence, instead of relying solely on  samples simulated at fixed parameter settings (as is the convention in standard Monte Carlo solutions), one can leverage machine learning tools and data simulated in the neighborhood of a parameter to improve estimates of quantities of interest and decrease the total number of simulated data points. Our contribution is three-fold:
\vspace{-1mm}
\begin{enumerate}
 \item   a new procedure for estimating the {\em likelihood ratio statistic},  which uses probabilistic classifiers and does not require repeated sampling at each $\theta$ or a separate interpolation or calibration step;
 \vspace{-1mm}
 \item an efficient procedure for estimating the {\em critical value} that guarantees valid tests and confidence sets, based on quantile regression without repeated sampling at each $\theta$;
 \vspace{-1mm}
 \item a new {\em goodness-of-fit technique} for computing empirical coverage of constructed confidence sets as a function of the unknown parameters $\theta$.
\end{enumerate}

Finally,  \texttt{ACORE} is simple and modular by construction. One can easily switch out, generalize or take pieces of the framework and apply it to any similar machine-learning-based LFI setting.
The theoretical results of Section~\ref{sec:hyp_test_via_or} hold for a general setting. In addition, given the vast arsenal of existing probabilistic classifiers developed in the literature, \texttt{ACORE} can be applied to many different types of complex data $\X$ (e.g., images, time series and functional data). In Section~\ref{sec: toy_example}, we show empirical results connecting the power of the constructed hypothesis tests to the performance of the classifier. 
Note also that Algorithms~\ref{alg:estimate_thresholds2} and \ref{alg:estimate_thresholds_conf2} for estimating parametrized critical values apply to any hypothesis test on $\theta$ of the form of Equation~\ref{eq:hypothesis_testing} for any test statistic $\tau$. The goodness-of-fit procedure in Section~\ref{sec:coverage_gof} for checking empirical coverage as a function of $\theta$ is also not tied to odds ratios.

\subsection{Related Work} 
The problem of constructing confidence intervals with good frequentist properties has a long history in statistics \citep{neyman1937inversion, Feldman1998UnifiedApproach, chaung2000hybridresampling}.
One of the earlier simulation-based approaches was developed in high energy physics (HEP) by \citet{Diggle1984HistogramEstimators}; their proposed scheme of estimating the likelihood and likelihoood ratio statistic nonparametrically by histograms of photon counts would later become a key component in the discovery of the Higgs Boson \citep{Aad2012HiggsBoson}. However, traditional approaches for building confidence regions and hypothesis tests in LFI rely on a  series of Monte Carlo samples at each parameter value $\theta$ 
\citep{barlow1993MC, weinzierl2000introduction, chad2009ConfidenceRegions}. 
Thus, these approaches quickly become inefficient with large or continuous parameter spaces. Traditional nonparametric approaches also have difficulties handling high-dimensional data without losing key information.

LFI has recently benefited from using powerful machine learning tools like deep learning to  estimate likelihood functions and likelihood ratios for complex data. 
Successful application areas include HEP \cite{guest2018deeplearningLHC}, astronomy \cite{Alsing2019delfi} and neuroscience \cite{Goncalves2019lfi_posterior}.
\texttt{ACORE} has some similarities to the work of
\citet{Cranmer2015LikRatio} which also uses machine learning
methods for frequentist inference in an LFI setting.
Other elements of \texttt{ACORE} such as leveraging the ability of ML algorithms to smooth over parameter space turning a density ratio estimate into a supervised classification problem have also previously been used in LFI settings: Works that smooth over parameter space include, e.g., Gaussian processes \citep{Frate2017GP_LFI, florent2018GP_LFI} and neural networks \cite{Baldi2016ParametrizedNN}. 
Works that turn a density ratio into a classification problem include applications to generative models (see  \citealt{shakir2016Implicit_models} for a review), and Bayesian LFI  \citep{Thomas2016likelihood,Gutmann2018LFI, Dinev2018DIREDensityRatio,Hermans}. Finally, like \texttt{ACORE}, \citet{thornton2017effective} explores frequentist guarantees of confidence regions; however those regions are built under a Bayesian framework.

\textbf{Novelty.} 
What distinguishes the \texttt{ACORE} approach from other  related work is that it uses an efficient procedure for estimating the (i) likelihood ratio, (ii) critical values and (iii) coverage of confidence sets across the entire parameter space, without the need for an extra  interpolation or calibration step (as in traditional Monte Carlo solutions and more recent ML approaches). To the best of our knowledge, (ii) and (iii) are entirely novel in the LFI literature. In contrast to other methods that estimate (i), \texttt{ACORE} does not make parametric assumptions or require an additional calibration step, and it can accommodate all types of hypotheses. We provide theoretical guarantees on our procedures in terms of power and validity (Section \ref{sec:hyp_test_via_or}; proofs in Supplementary Material~\ref{appendix:proofs}). We also offer a scheme for how to choose ML algorithms and the number of simulations so as to have good power properties and valid inference in practice.

{\bf Notation.}
Let $F_{\theta}$ with density $f_\theta$  represent the stochastic forward simulator for a sample point $\X \in \mathcal{X}$ at parameter $\theta \in \Theta$.  We denote  i.i.d ``observable'' data from $F_{\theta}$ by $\D=\left\{ \X_1^{\text{obs}},\ldots,\X_n^{\text{obs}} \right\}$, and the actually observed or measured data by $D=\left\{ \x_1^{\text{obs}},\ldots,\x_n^{\text{obs}} \right\}$. The likelihood function $\L(\D;\theta)=\prod_{i=1}^{n} f_\theta(\X_i^{\text{obs}})$.

\section{Statistical Inference in a Traditional Setting}~\label{sec::trad_inference}

\vspace{-6mm}

We begin by reviewing elements of traditional statistical inference that play a key role in \texttt{ACORE}.

{\bf Equivalence of tests and confidence sets.} A classical approach to constructing a confidence set for an unknown parameter $\theta \in \Theta$ 
is to  
invert a series of hypothesis tests \citep{neyman1937inversion}: Suppose that for each possible value $\theta_0 \in \Theta$, there is a level $\alpha$ test $\delta_{\theta_0}$ of  
\begin{equation}
H_{0, \theta_0}: \theta=\theta_0 \ \ \mbox{versus} \ \ H_{1, \theta_0}: \theta \neq \theta_0;
\label{eq:test_Neyman}
\end{equation}
that is, a test $\delta_{\theta_0}$  where the type I error (the probability of erroneously rejecting a true null hypothesis $H_{0, \theta_0}$) is no larger than $\alpha$.   For observed data $\D=D$, now define $R(D)$ as the  set of all parameter values $\theta_0 \in \Theta$ for which the test $\delta_{\theta_0}$  does not reject $H_{0,\theta_0}$.
Then, by construction, the random set $R(\D)$ satisfies
$$\P \left[\theta_0 \in R(\D) \ \middle\vert \ \theta=\theta_0\right] \geq 1-\alpha$$
for all $\theta_0 \in \Theta$. That is,  $R(\D)$ defines a $(1-\alpha)$ {\em confidence set} for $\theta$. Similarly, we can define a test with a desired significance level  from a confidence set with a certain coverage.
 
{\bf Likelihood ratio test.} A general form of hypothesis tests that often leads to high power is the likelihood ratio test (LRT). Consider testing 
\begin{equation}   \label{eq:hypothesis_testing}
 H_0: \theta \in \Theta_0 \ \ \mbox{versus} \ \ H_1: \theta \in \Theta_1,
\end{equation}
where $\Theta_1=\Theta \setminus \Theta_0$.  For the {\em likelihood ratio (LR) statistic}, 
\begin{equation}
\label{eq::LRT}
\Lambda(\D; \Theta_0) =  
\log \frac{\sup_{\theta \in \Theta_0}\L(\D;\theta)}{\sup_{\theta \in \Theta}\L(\D;\theta)},
\end{equation}
the LRT of hypotheses (\ref{eq:hypothesis_testing})
rejects $H_0$ when $\Lambda(D; \Theta_0) < C$ for some constant $C$.

Figure~\ref{fig:neyman_inversion_statistics} illustrates the construction of confidence sets for $\theta$ from level $\alpha$ likelihood ratio tests (\ref{eq:test_Neyman}). 
The critical value for each such test $\delta_{\theta_0}$ is 
$C_{\theta_0} =  \left\{C: \P \left[\Lambda(\D; \theta_0)< C \mid \theta=\theta_0 \right] = \alpha \right\}.$

\begin{figure}[!ht]
    \centering
    \includegraphics[width=0.48\textwidth]{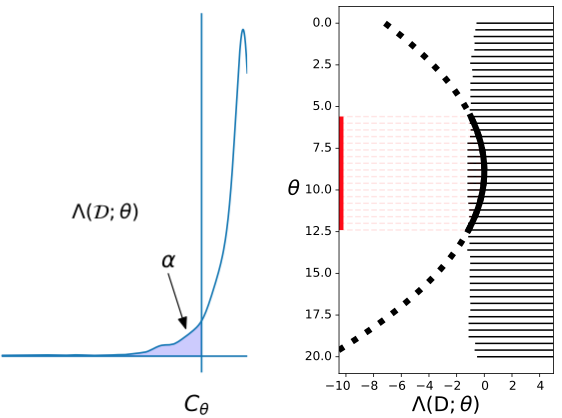}
    \vspace*{-6mm}
    \caption{Constructing confidence intervals from hypothesis tests. {\em Left:} For each $\theta \in \Theta$, we find the critical value $C_{\theta}$ that rejects the null hypothesis $H_{0,\theta}$ at level $\alpha$; that is,   $C_{\theta}$ is the $\alpha$-quantile of the distribution of the likelihood ratio statistic  $\Lambda(\mathcal{D}; \theta)$ under the null. {\em Right:} The horizontal lines represent the acceptance region for each $\theta \in \Theta$. Suppose we observe data $\D=D$. The confidence set for $\theta$ (indicated with the red line) consists of all $\theta$-values for which the observed test statistic $\Lambda(D; \theta)$ (indicated with the black curve) falls in the acceptance region.
    }
    \label{fig:neyman_inversion_statistics}
\end{figure}

\begin{figure}[!ht]
    \centering
    \includegraphics[width=0.49\textwidth]{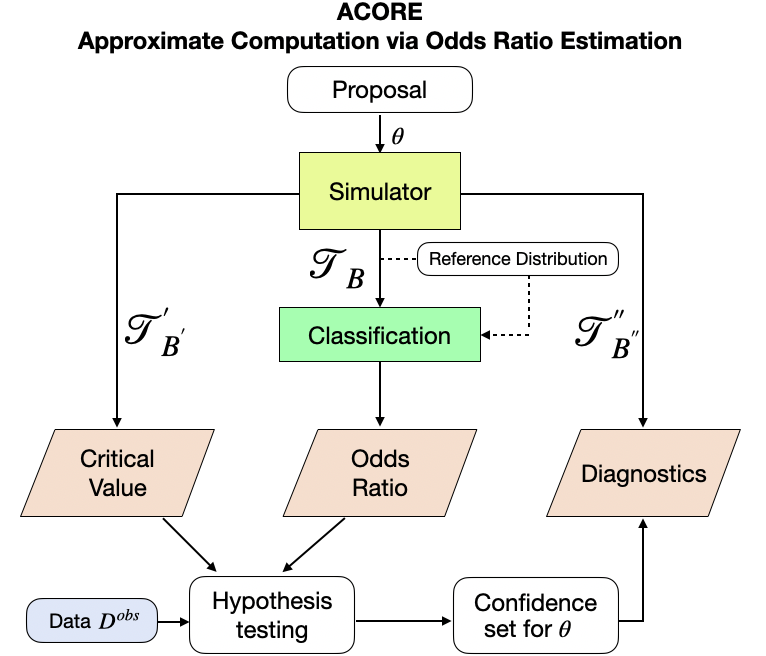}
    \vspace*{-5mm}
    \caption{
    Schematic diagram of \texttt{ACORE}. The simulator provides synthetic observable data $\T_B$ for learning a parametrized odds ratio via probabilistic classification.  The simulator also generates a separate sample $\T'_{B'}$ for learning critical values as a function of $\theta \in \Theta$. Once data $\D^{\rm obs}$ are observed, the odds ratio can be used to construct hypothesis tests or confidence sets for $\theta$. \texttt{ACORE}  provides diagnostics for computing the empirical coverage of constructed confidence sets as a function of the (unknown) parameter $\theta$. The three main parts of \texttt{ACORE} (critical value, odds ratio, diagnostics) are separate modules. Each module leverages machine learning methods in the training phase and is amortized, i.e., they perform inference on new data without having to be retrained.
    }
    \label{fig:work_structure}
\end{figure}

\section{ \texttt{ACORE}: Approximate Computation via Odds Ratio Estimation} 
In a likelihood-free inference setting, we cannot directly evaluate the likelihood ratio  statistic. Here we describe the details of how a simulation-based approach (\texttt{ACORE}, Figure~\ref{fig:work_structure}) can lead to hypothesis tests and confidence sets with good frequentist properties.

\subsection{Hypothesis Testing via Odds Ratios} \label{sec:hyp_test_via_or}
We start by simulating a labeled sample for computing odds ratios. The estimated odds ratio then defines a new test statistic that we use in place of the unknown likelihood ratio statistic.

{\bf Simulating a labeled sample.}
Let $G$ be a distribution with larger support than  $F_\theta$ for all $\theta \in \Theta$. The distribution $G$ could for example be a  dominating distribution which it is easy to sample from.
We use  $F_{\theta}$ and $G$ to simulate a labeled training sample $T_B = \{\theta_i,\x_i,y_i\}_{i=1}^B$ for estimating odds ratios. The random sample  $\T_B= \{\theta_i,\X_i,Y_i\}_{i=1}^B$ is identically distributed as $(\theta,\X,Y)$, where  the parameters $\theta \sim r_{\Theta}$ (a fixed proposal distribution over $\Theta$), the ``label''  $Y \sim \text{Ber}(p)$ (a Bernoulli distribution with known $p$ with $Y$ independent of $\theta$), $\X|\theta,Y=1\sim F_{\theta}$ and $\X|\theta,Y=0\sim G$. That is, the label $Y_i$ is the indicator that the sample point $\X_i$  was generated from $F_\theta$ rather than $G$. We  call $G$  a ``reference distribution'' as we are comparing $F_\theta$ for different $\theta$ with this distribution. For all our experiments in this work we use p=1/2; other choices could account for computational differences in sampling from $F_\theta$ versus $G$. (Algorithm \ref{alg:joint_y} in Supplementary Material~\ref{appendix::simulation_training} summarizes our procedure.)

{\bf Odds ratios.}
For fixed $\x$, we define the {\em odds} at $\theta$ as $${\O}(\x;\theta):=\frac{\P(Y=1|\theta,\x)}{\P(Y=0|\theta,\x)},$$ 
and the {\em odds ratio} at $\theta_0,\theta_1 \in \Theta$ as
$$\Or(\x;\theta_0,\theta_1):= \frac{\O(\theta_0;\x)}{\O(\theta_1;\x)}.$$ 
One way of interpreting the odds ${\O}(\theta,\X)$ is to regard it as a measure of the chance that $\X$ was generated from $F_\theta$. That is, a large odds ${\O}(\theta,\x)$ reflects the
fact that it is plausible that $\x$ was generated from $F_\theta$ (rather than $G$).
Thus, $\Or(\x;\theta_0,\theta_1)$ measures the  plausibility that $\x$
was generated from $\theta_0$ rather than $\theta_1$.
When testing (\ref{eq:hypothesis_testing}), we therefore reject $H_0$ if    $\sup_{\theta_0 \in \Theta_0}\inf_{\theta_1 \in \Theta} \sum_{i=1}^n  \log \left( \Or(\X_i^{\text{obs}};\theta_0,\theta_1)\right) < C$, for some constant $C$.
By Bayes rule, this is just the likelihood ratio test of (\ref{eq:hypothesis_testing}).

{\bf Hypothesis testing in an LFI setting.}
In an LFI setting, we cannot directly evaluate the likelihood ratio statistic (\ref{eq::LRT}). The advantage of rewriting the LRT in terms of odds ratios is that we can forward-simulate a labeled training sample $\mathcal{D}_B$, as described above, and then use a probabilistic classifier (suitable for the data at hand) to efficiently estimate the odds ratios $\Or(\x;\theta_0,\theta_1)$ for all $\theta_1, \theta_2 \in \Theta$:
The probabilistic classifier compares data from the forward simulator $F_\theta$ with data from the reference distribution $G$ and returns a parametrized odds estimate $\widehat{\O}(\x; \theta)$, which is a function of $\theta \in \Theta$. We can directly compute the 
odds ratio estimate $\widehat{\Or}(\x; \theta_0,\theta_1)$
at any two values $\theta_1, \theta_2 \in \Theta$  
from $\widehat{\O}(\x; \theta)$. There is no need for a separate training step. 
 
We reject $H_0$ if the \texttt{ACORE} test statistic defined as 
\begin{align}
\label{eq:odds_ratio_statistic}
 \tau(\D; \Theta_0):=\sup_{\theta_0 \in \Theta_0} \ \inf_{\theta_1 \in \Theta} \sum_{i=1}^n  \log \left( \widehat{\Or}(\X_i^{\text{obs}};\theta_0,\theta_1)\right)
\end{align}
is small enough for observed data $\D=D$.
If the probabilities learned by the classifier 
are well estimated, $\tau$
is exactly the likelihood ratio statistic:
\begin{Prop}[Fisher Consistency] \label{prop::consistency}
If $\ \widehat{\P}(Y=1|\theta,\x)=\P(Y=1|\theta,\x)$ for every $\theta$ and $\x$, 
then the \texttt{ACORE} test statistic (\ref{eq:odds_ratio_statistic}) is the likelihood ratio statistic
(Equation~\ref{eq::LRT}).
\end{Prop}

{\bf Estimating the critical value.} 	A key question is how to {\em efficiently} estimate the critical value of a test. In this section we consider a single composite null hypothesis  $H_0:\theta \in \Theta_0$. (The setting for constructing confidence sets by testing (\ref{eq:test_Neyman}) for all $\theta_0 \in \Theta$ is discussed in Section~\ref{sec::confidence_sets}.). 
Suppose that we reject the null hypothesis  
if the test statistic (\ref{eq:odds_ratio_statistic})
is smaller than some constant $C$. To achieve a test with a desired level of significance  $\alpha$, we need (for maximum power) the largest  $C$
that satisfies 
\begin{equation} \label{eq:significance_level}
\sup_{\theta \in \Theta_0} \P \left(\tau(\D; \Theta_0) < C \mid \theta \right) \leq \alpha .
\end{equation}
However, we cannot explicitly compute the critical value $C$ or the rejection probability as we do not know the distribution of the test statistic $\tau$.

Simulation-based approaches are often used to compute rejection probabilities and critical values in lieu of large-sample theory approximations. Typically, such simulations compute 
a separate Monte Carlo simulation at each fixed $\theta \in \Theta_0$ on, e.g., a fine enough grid on $\theta$. That is, the convention is to rely solely on sample points generated at fixed $\theta$ to estimate the rejection probabilities $\P(\tau(\mathcal{D};\Theta_0)< C|\theta)$. Here we propose to estimate the critical values $C$ for all  $\theta \in \Theta_0$ and significance levels $\alpha \in [0,1]$ simultaneously. 
At the heart of our approach is the key observation that the rejection probability  $\P(\tau(\mathcal{D};\Theta_0)< C|\theta)$ is a conditional cumulative distribution function, which in many settings varies smoothly as a  function of $\theta$ and $C$. Thus, similar to how we estimate odds for the \texttt{ACORE} statistic, 
one can use data generated in the neighborhood of $\theta$ to improve estimates of our quantities of interest at any $\theta$. This is what a quantile regression implicitly does to estimate $C$.

Algorithm \ref{alg:estimate_thresholds2} outlines the details of the procedure for estimating $C$.  In brief, we use a training sample $\T'_{B'}=\{(\theta_i, \tau_i) \}_{i=1}^{B'}$ (independent of $\T_B$) to estimate the $\alpha$-conditional quantile  $c_\alpha(\theta)$ defined by $\P \left(\tau \leq c_\alpha(\theta) \mid \theta \right) = \alpha.$   Let $\widehat{c}_\alpha(\theta)$ be the estimate of $c_\alpha(\theta)$ from a quantile regression of  $\tau$ on $\theta$. By (\ref{eq:significance_level}), our estimate of the critical value $C$ is $
 \widehat{C} =  
 \inf_{\theta \in \Theta_0} \widehat{c}_\alpha(\theta).$
As we shall see, even if the odds are not well estimated, tests and confidence regions based on estimated odds are still valid as long as the thresholds are well estimated. Next  we show that the sample size $B'$ in Algorithm~\ref{alg:estimate_thresholds2}  controls the type I error (Theorem~\ref{thm:convergenceCutoffs}), whereas the training sample size $B$ for estimating odds  is related to the power of the test (Theorem~\ref{thm:convergenceLRT}).

\begin{algorithm}[!ht]
	\caption{
	Estimate the critical value $C$ for a level-$\alpha$ test of composite hypotheses $H_0: \theta \in \Theta_0$ vs. $H_1: \theta \in \Theta_1$
	 	\small
	 	}\label{alg:estimate_thresholds2}
	\algorithmicrequire \ {\small  stochastic forward simulator $F_\theta$; sample size $B'$ for training quantile regression estimator; $r_{\Theta_0}$ (a  fixed proposal distribution over the null region $\Theta_0$); test statistic $\tau$; quantile regression estimator; desired level $\alpha \in (0,1)$}\\
	\algorithmicensure \ {\small estimated critical value  $\widehat{C}$}
	\begin{algorithmic}[1]
	\STATE Set $\T' \gets \emptyset$
\FOR{i in \{1,\ldots,B'\}} 
\STATE  Draw parameter $\theta_i \sim r_{\Theta_0}$
\STATE Draw sample $\X_{i,1},\ldots,\X_{i,n}  \stackrel{iid}{\sim}  F_{\theta_i}$
\STATE Compute  test statistic $\tau_i \gets \tau((\X_{i,1},\ldots,\X_{i,n});\Theta_0) $
\STATE  $\T' \gets \T' \cup  \{(\theta_i,\tau_{i})\}$
\ENDFOR
\STATE  Use $\T'$ to learn parametrized function   
$\widehat{c}_\alpha(\theta) := \widehat{F}_{\tau|\theta}^{-1}(\alpha | \theta)$ 
via quantile regression of  $\tau$ on $\theta$

\textbf{return}  $ \widehat{C} \leftarrow
 \inf_{\theta \in \Theta_0} \widehat{c}_\alpha(\theta)$

	\end{algorithmic}
\end{algorithm}

{\bf Theoretical guarantees.}
We denote convergence in probability and in distribution by $ \overset{\P}{\rightarrow}$ and $\xrightarrow[]{\enskip \small \mbox{Dist} \enskip}$, respectively.
We start by showing that our procedure leads to valid hypothesis tests (that is,  tests that control the type I error probability) as long as $B'$ in Algorithm~\ref{alg:estimate_thresholds2} is large enough. In order to do so,
we assume that the quantile regression estimator used in Algorithm~\ref{alg:estimate_thresholds2} to estimate the
critical values is consistent in the following sense:
\begin{assumption}
\label{assum:quantile_consistent}
Let $\hat F_{B'}(\cdot|\theta)$ be the estimated cumulative distribution function of the test statistic $\tau$ conditional on $\theta$ based on a sample size $B'$, and
let $ F(\cdot|\theta)$ be true conditional distribution.
For every $\theta \in \Theta_0$, assume that the quantile regression estimator is such that
$$\sup_{t \in \mathbb{R}}|\hat F_{B'}(t|\theta)- F(t|\theta)|\xrightarrow[B' \longrightarrow\infty]{\enskip \P \enskip} 0$$
\end{assumption}

Under some conditions, Assumption~\ref{assum:quantile_consistent} holds for instance for quantile regression forests \citep{meinshausen2006quantile}.

Next we show that, for every fixed training sample size  $B$ in Algorithm~\ref{alg:joint_y},  
Algorithm~\ref{alg:estimate_thresholds2} yields a valid hypothesis test as $B' \rightarrow \infty$. The result holds even if the likelihood ratio statistic is not well estimated.

\begin{theorem}
 \label{thm:convergenceCutoffs}
Let 
$C_{B,B'} \in \mathbb{R}$ be the 
critical value 
of the test based on
the statistic $\tau=\tau_B$  for a training  sample size $B$ with critical value chosen according to Algorithm~\ref{alg:estimate_thresholds2}
for a fixed $\alpha \in (0,1)$. If the quantile
estimator satisfies Assumption~\ref{assum:quantile_consistent}
and $|\Theta|<\infty$, then 
$$ C_{B,B'}  \xrightarrow[B' \longrightarrow\infty]{\enskip \P \enskip} C_B^*,$$
where $C_B^*$ is such that 
$$\sup_{\theta \in \Theta_0}\P(\tau_B \leq C_B^*|\theta) = \alpha.$$
\end{theorem}
 
Finally we show that as long as the probabilistic classifier is consistent and the critical values are well estimated (which holds for large $B'$ according to Theorem~\ref{thm:convergenceCutoffs}), the power of the \texttt{ACORE} test converges to the power of the LRT as $B$ grows.
 
\begin{theorem}
\label{thm:convergenceLRT}
Let $\widehat{\phi}_{B,C_B}(\mathcal{D})$ be the test based on
the statistic $\tau=\tau_B$  for a labeled sample size $B$ with critical value $C_B \in \mathbb{R}$.\footnote{That is, $\widehat{\phi}_{B,C_B}(\mathcal{D})=1  \iff \tau_B(\D; \Theta_0) < C_B$.} Moreover, 
let 
$\phi_{C^*}(\mathcal{D})$ be the likelihood ratio test with  critical value
$C^* \in \mathbb{R}$.\footnote{That is, $\phi_{C^*}(\mathcal{D})=1  \iff \Lambda(\D; \Theta_0) < C^*$.}
If, for every
$\theta \in \Theta$,
$$\widehat{\P}(Y=1|\theta,\X) \xrightarrow[B \longrightarrow\infty]{\enskip \P \enskip} \P(Y=1|\theta,\X),$$
where $|\Theta|<\infty$, and $\widehat{C}_B$
is such that
$\widehat{C}_B
\xrightarrow[B \longrightarrow\infty]{\enskip \mbox{Dist} \enskip}C^*$,
then, for every $\theta \in \Theta$,
$$\P\left(\widehat{\phi}_{B,\widehat{C}_B}(\mathcal{D})=1|\theta \right) 
\xrightarrow[B \longrightarrow\infty]{} \P\left(\phi_{C^*}(\mathcal{D})=1|\theta\right).
$$
\end{theorem}

\subsection{Confidence Sets}  \label{sec::confidence_sets}
To construct a confidence set for $\theta$, we use the equivalence of tests and confidence sets (Section~\ref{sec::trad_inference}): Suppose that we for every $\theta_0 \in \Theta$ can find the critical value $C_{\theta_0}$ of  a test of (\ref{eq:test_Neyman}) with type I error no larger than $\alpha$. The random set
$$R(\D) = \left\{ \theta_0  \in \Theta \ \middle\vert  \tau(\D;\theta_0) \geq C_{\theta_0}\right\},$$
then defines a $(1-\alpha)$ confidence region for $\theta$.

However, rather than repeatedly running Algorithm~\ref{alg:estimate_thresholds2} for each null hypothesis $\Theta_0=\{\theta_0\}$ separately, we estimate all critical values $C_{\theta_0}$ (for different $\theta_0 \in \Theta$) simultaneously. Algorithm~\ref{alg:estimate_thresholds_conf2} outlines our procedure. Again, we use quantile regression to learn a parametrized function $C_{\theta_0}$. The whole procedure for computing confidence sets via \texttt{ACORE} is  summarized in  Algorithm~\ref{alg:conf_reg} in Supplementary Material~\ref{appendix::confidence_set}. Theorem~\ref{thm:convergenceCutoffs} implies that the constructed confidence set has the nominal $1-\alpha$ confidence level 
as $B' \rightarrow \infty$. The size of the confidence set depends on the training sample size $B$ and the classifier.

\begin{algorithm}[!ht]
	\caption{[Many Simple Null Hypotheses]
	Estimate the critical values $C_{\theta_0}$  for a level-$\alpha$ test of  $H_{0, \theta_0}: \theta = \theta_0$ vs. $H_{1, \theta_0}: \theta \neq \theta_0$
	for all $\theta_0 \in \Theta$ simultaneously  
	} 
	\label{alg:estimate_thresholds_conf2} 
	\algorithmicrequire \ {\small  stochastic forward simulator $F_\theta$; sample size $B'$ for training quantile regression estimator; $r$ (a fixed proposal distribution over the full parameter space $\Theta$); test statistic $\tau$; quantile regression estimator;  desired level $\alpha \in (0,1)$}\\
	\algorithmicensure \ {\small estimated critical values  $\widehat{C}_{\theta}$ for all $\theta=\theta_0 \in \Theta$}
	\begin{algorithmic}[1]
	\STATE Set $\mathcal{D'} \gets \emptyset$ 
\FOR{i in \{1,\ldots,B'\}}
\STATE Draw parameter $\theta_i \sim r$
\STATE Draw sample 
$\X_{i,1},\ldots,\X_{i,n}  \stackrel{iid}{\sim}  F_{\theta_i}$
\STATE Compute test statistic $\tau_i \gets \tau((\X_{i,1},\ldots,\X_{i,n});\theta_i) $
\STATE $\mathcal{D'} \gets \mathcal{D'} \cup  \{(\theta_i,\tau_{i})\}$  
\ENDFOR
\STATE  Use $\mathcal{D'}$ to learn parametrized function   $\widehat{C}_\theta := \widehat{F}_{\tau|\theta}^{-1}(\alpha | \theta)$
via quantile regression of  $\tau$ on $\theta$
	\textbf{return}
	 $\widehat{C}_{\theta_0} \leftarrow \widehat{F}_{\tau|\theta_0}^{-1}(\alpha | \theta_0)$
	\end{algorithmic}
\end{algorithm}

\subsection{Evaluating Empirical Coverage for All Possible Values of $\theta$} \label{sec:coverage_gof}
After the parametrized \texttt{ACORE} statistic and the critical values  have been estimated, it is important to check whether the 
resulting confidence sets  indeed are valid or, equivalently, if the resulting hypothesis tests have the nominal significance level. We also want to identify regions in parameter space where we clearly overcover. That is, the two main questions are: (i) do the constructed confidence sets satisfy
$$\P \left[\theta_0 \in R(\D) \ \middle\vert \ \theta=\theta_0\right] \geq 1-\alpha,$$
for every $\theta_0 \in \Theta$, and (ii) how close is the actual coverage to the nominal  
 confidence level 
$1-\alpha$? To answer these questions, we propose a goodness-of-fit procedure where we draw $B''$ new samples from the simulator given $\theta$,  construct a confidence set for each sample, and then check which computed regions include the ``true'' $\theta$. More specifically: we generate
a  set $\mathcal{T}_{B''}'' =\{(\theta'_1,\D'_1),\ldots,(\theta'_{B''},\D'_{B''})\}$, where $\theta_i' \sim r_{\Theta}$ and
$\D'_i$
is a sample of size $n$
of i.i.d. 
observable data from
$F_{\theta'_i}$. We then define
$$W_i:=\I\left(\theta'_i \in R(\D'_i)\right),$$
where $R(\D'_i)$ is the  confidence set for $\theta$
for data $\D'_i$.
If $R$
has the correct coverage, then
\begin{equation*}
\P(W_i = 1|\theta_i )\geq 1-\alpha.
\end{equation*}
We can estimate the probability $\P(W_i = 1|\theta_i )$ using any probabilistic classifier; some methods also provide confidence bands that assess the uncertainty in estimating this quantity \citep{eubank1993confidence,claeskens2003bootstrap,krivobokova2010simultaneous}. By comparing the estimated probability to $1-\alpha$, we have a  diagnostic tool for checking how close we are to the nominal  confidence level over the entire parameter space $\Theta$.
See Figure~\ref{fig:toy_example_coverage} for an example.

Finally note that our procedure parametrizes the coverage of the confidence set as a function of the true parameter value. This is in contrast to other goodness-of-fit techniques (e.g., \citealt{cooks2006posteriorquantile, bordoloi2010photo, talts2018validating, desc_photoz})
that only check for \emph{marginal} coverage, i.e., $n^{-1}\sum_{i=1}^{n} W_i \geq 1-\alpha$.

\section{Toy Examples}\label{sec: toy_example}

We consider two examples where the true likelihood is known. In the first example, the forward simulator $F_\theta$ follows a $\Pois(100 + \theta)$ distribution similar to the signal-background model in Section~\ref{sec: applications}. In the second example, we consider a Gaussian mixture model (GMM) with two unit-variance Gaussians centered at $-\theta$ and $\theta$, respectively. In both examples, $n=10$, the proposal distribution $r_{\Theta}$ is a uniform distribution, and the reference distribution $G$ is a normal distribution. Table~\ref{tab:toy_example_setup} summarizes the set-up. 

\begin{table}[!ht]
\resizebox{0.475\textwidth}{!}{%
\begin{tabular}{|c||c|c|}
\hline
 & \textit{Poisson Example} & \textit{GMM Example} \\ \hline
 $r_{\Theta}$ & ${\rm Unif}(0,20)$ & ${\rm Unif}(0,10)$ \\ \hline
$F_\theta$ & $\Pois(100 + \theta)$ & $\frac{1}{2}\mathcal{N}(-\theta, 1) + \frac{1}{2}\mathcal{N}(\theta, 1)$ \\ \hline
$G$ & $\mathcal{N}(110, 15^2)$ & $\mathcal{N}(0, 5^2)$ \\ \hline
True $\theta$ & $\theta_0 = 10$ & $\theta_0 = 5$\\ \hline
\end{tabular}}
\vspace*{-2.5 mm}
\caption{Set-up for the two toy examples.
}\label{tab:toy_example_setup}
\end{table}

\begin{table*}[!ht]
\resizebox{0.5\textwidth}{!}{%
\begin{tabular}{|c|c|c|c|c|}
\multicolumn{5}{c}{Poisson Example} \\ \hline
\textit{$B$} & \textit{Classifier} & \textit{Cross} & \textit{Average} & \textit{Size of} \\
 & & \textit{Entropy Loss} & \textit{Power} & \textit{Confidence Set [\%]} \\ \hline
\multirow{3}{*}{100} & MLP & 0.87 $\pm$ 0.27 & 0.24 & 75.9 $\pm$ 19.3\\
 & NN & 0.76 $\pm$ 0.15 & 0.29 & 71.6 $\pm$ 19.7\\
 & QDA & \textbf{0.66 $\pm$ 0.02} & \textbf{0.41} & \textbf{60.0 $\pm$ 15.6}\\ \hline
\multirow{3}{*}{500} & MLP & 0.69 $\pm$ 0.01 & 0.35 & \textbf{65.9 $\pm$ 20.4}\\
 & NN & 0.67 $\pm$ 0.01 & 0.38 & \textbf{62.9 $\pm$ 15.8}\\
 & QDA & \textbf{0.64 $\pm$ 0.01} & \textbf{0.47} & \textbf{54.2 $\pm$ 9.4}\\ \hline
\multirow{3}{*}{1,000} & MLP & 0.69 $\pm$ 0.01 & 0.37 & \textbf{63.3 $\pm$ 19.8}\\
 & NN & 0.66 $\pm$ 0.01 & 0.44 & \textbf{56.9 $\pm$ 15.9}\\
 & QDA & \textbf{0.64 $\pm$ 0.01} & \textbf{0.50} & \textbf{51.3 $\pm$ 7.7}\\ \hline
 \hline
- & {\color{red} \textbf{Exact}} & {\color{red} \textbf{0.64 $\pm$ 0.01}} &  {\color{red} \textbf{0.54}}  & {\color{red} \textbf{45.0 $\pm$ 4.9}} \\ \hline
\end{tabular} \\ [5ex]
}
\resizebox{0.5\textwidth}{!}{%
\begin{tabular}{|c|c|c|c|c|}
\multicolumn{5}{c}{GMM Example} \\ \hline
\textit{$B$} & \textit{Classifier} & \textit{Cross} & \textit{Average} & \textit{Size of}  \\
 & & \textit{Entropy Loss} & \textit{Power} & \textit{Confidence Set [\%]} \\ \hline
\multirow{3}{*}{100} & MLP & \textbf{0.39 $\pm$ 0.03} & \textbf{0.88} & \textbf{14.1 $\pm$ 4.7}\\
 & NN & 0.81 $\pm$ 0.31 & 0.42 & 58.4 $\pm$ 23.3\\
 & QDA & 0.64 $\pm$ 0.02 & 0.15 & 85.3 $\pm$ 21.1\\ \hline
\multirow{3}{*}{500} & MLP & \textbf{0.35 $\pm$ 0.01} & \textbf{0.90} & \textbf{12.1 $\pm$ 2.4}\\
 & NN & 0.45 $\pm$ 0.05 & 0.57 & 44.3 $\pm$ 24.1\\
 & QDA & 0.62 $\pm$ 0.01 & 0.15 & 84.9 $\pm$ 19.9\\ \hline
\multirow{3}{*}{1,000} & MLP & \textbf{0.35 $\pm$ 0.01} & \textbf{0.90} & \textbf{12.1 $\pm$ 2.5}\\
 & NN & 0.41 $\pm$ 0.02 & 0.77 & \textbf{24.9 $\pm$ 15.9}\\
 & QDA & 0.62 $\pm$ 0.01 & 0.12 & 88.1 $\pm$ 18.0\\ \hline
 \hline
- & {\color{red} \textbf{Exact}} & {\color{red} \textbf{0.35 $\pm$ 0.01}} &  {\color{red} \textbf{0.92}}  & {\color{red} \textbf{9.5 $\pm$ 2.0}} \\ \hline
\end{tabular}
}
\caption{
Results for Poisson example (left) and GMM example (right). The tables show the cross entropy loss, power (averaged over $\theta$) and size of \texttt{ACORE} confidence sets for different values of $B$ and for different classifiers. These results are based on 100 repetitions; the numbers represent the mean and one standard deviation. 
The best results in each setting are marked in bold-faced; we see that the classifier with the lowest cross entropy loss (a quantity that is easily computed in practice) is linked with the highest average power and the smallest confidence set. As $B$ increases, the best \texttt{ACORE} values approach the values for the exact LRT, listed in the bottom row in red color. (The QDA for the GMM example does not improve with increasing $B$ because  the quadratic classifier cannot separate $F_\theta$ and $G$ in a mixed distribution with three modes, hence breaking the assumption of Theorem~\ref{thm:convergenceLRT}.) All nine probabilistic classifiers yield valid 90\% confidence regions according to our diagnostics; see Table~\ref{tab:supp_mat_toy_example_coverage}.
}
\label{tab:cross_ent_poisson}
\end{table*}

First we investigate how the power of \texttt{ACORE} and the size of the derived confidence sets depend on the performance of the classifier used in the odds ratio estimation (Section~\ref{sec:hyp_test_via_or}). We consider three classifiers: multilayer perceptron (MLP), nearest neighbor (NN) and quadratic discriminant analysis (QDA).
For different values of $B$ (sample size for estimating odds ratios), we compute the binary cross entropy (a measure of classifier performance),  the power as a function of $\theta$, and the size of the constructed confidence set. Table~\ref{tab:cross_ent_poisson} summarizes results based on 100 repetitions. (To compute the critical values in Algorithm~\ref{alg:estimate_thresholds_conf2}, we use quantile gradient boosted trees and a large enough sample size 
$B'=5000$ to guarantee $90\%$ confidence sets; see Supplementary Material~\ref{sec:toy_examples_supp_mat}.)
The last row of the table shows the best attainable cross entropy loss (Supplementary Material~\ref{sec:cross_entropy_loss}), the confidence set size and power for the true likelihood function. For all 18 settings, the computation of one \texttt{ACORE} confidence set takes between 10 to 30 seconds on a single CPU.\footnote{More specifically, an 8-core Intel Xeon 3.33GHz X5680 CPU.} A full breakdown of the runtime of \texttt{ACORE} confidence sets can be found in Supplementary Material~\ref{supp_mat: computational_costs}.

For each setting with fixed $B$, the best classifier according to cross entropy loss achieves the highest power and the smallest confidence set.\footnote{In traditional settings, high power has been shown to lead to a small expected interval size under certain distributional  assumptions~\citep{pratt1961intervals,Ghosh1961intervals}.} 
Moreover, as B increases, the best values (marked in bold-faced) get closer to those of the true likelihood (marked in red). The cross-entropy loss is easy to compute in practice. Our results indicate that minimizing the cross-entropy loss is a good rule of thumb for achieving \texttt{ACORE} inference results with desirable  statistical properties.

\begin{figure}[!ht]
    \centering
    \includegraphics[width=0.5\textwidth]{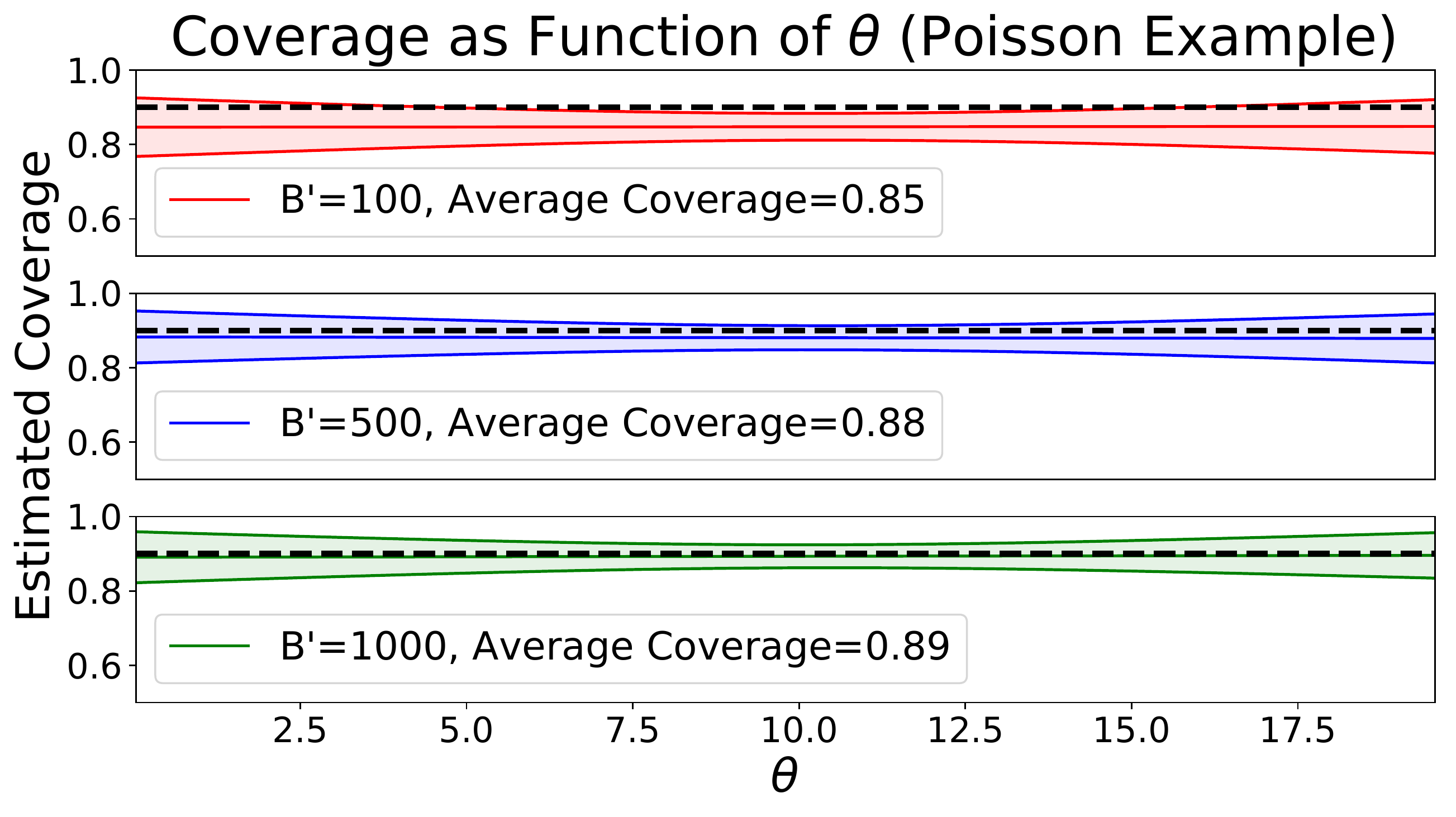}
    \vspace*{-7mm}
    \caption{Estimated coverage as a function of $\theta$ in the Poisson example for \texttt{ACORE} with different values of $B'$. The mean and one standard deviation prediction intervals are estimated via logistic regression. Our diagnostics show that $B'=500$ is large enough to achieve the nominal confidence level $1-\alpha=0.9$. (We here use $n=10$, a QDA 
    classifier with $B=1000$ and gradient boosted quantile regression).
    }
    \label{fig:toy_example_coverage}
\end{figure}

\begin{figure*}[!ht]
    \centering
    \includegraphics[width=0.33\textwidth]{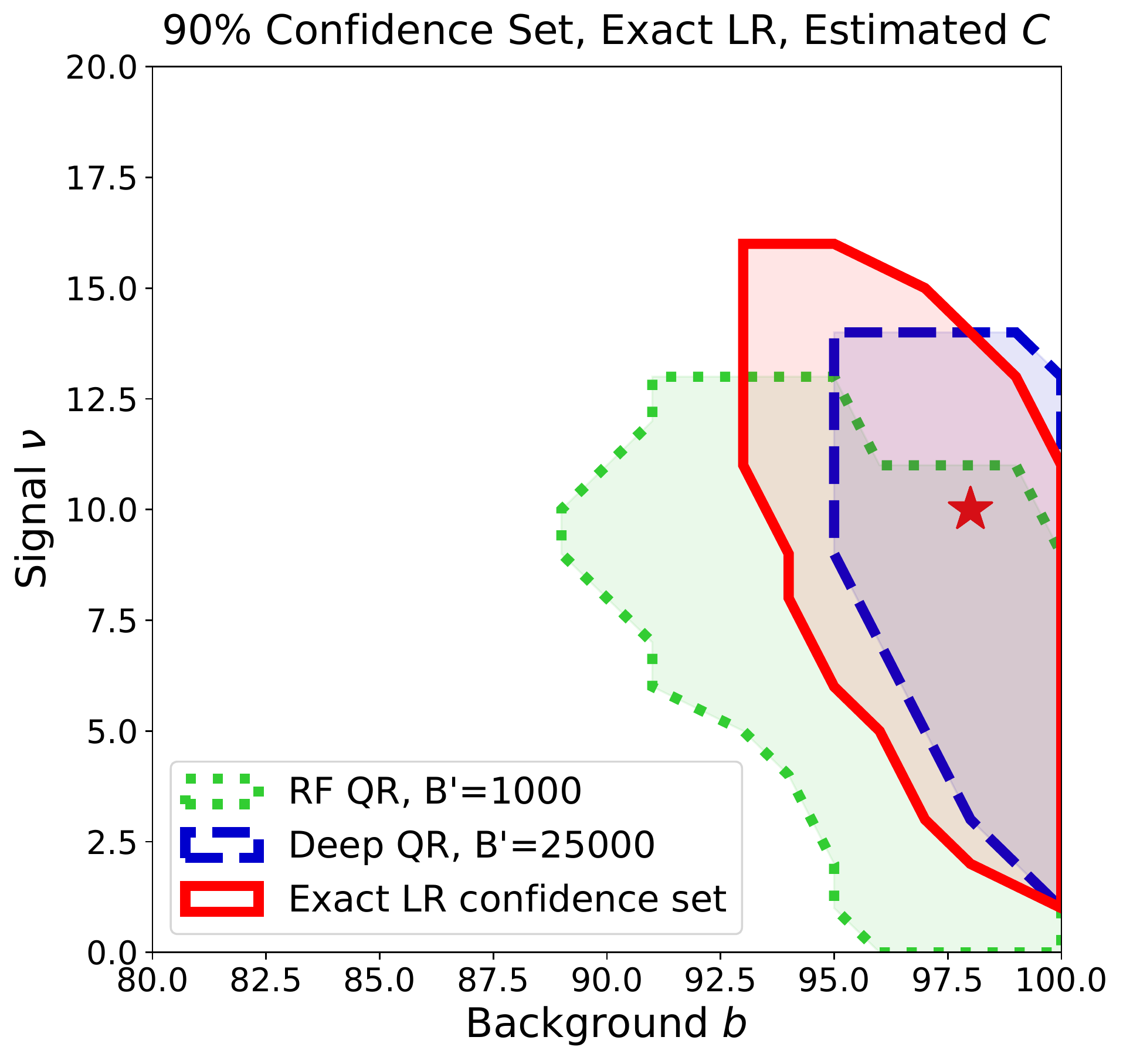} 
    \includegraphics[width=0.33\textwidth]{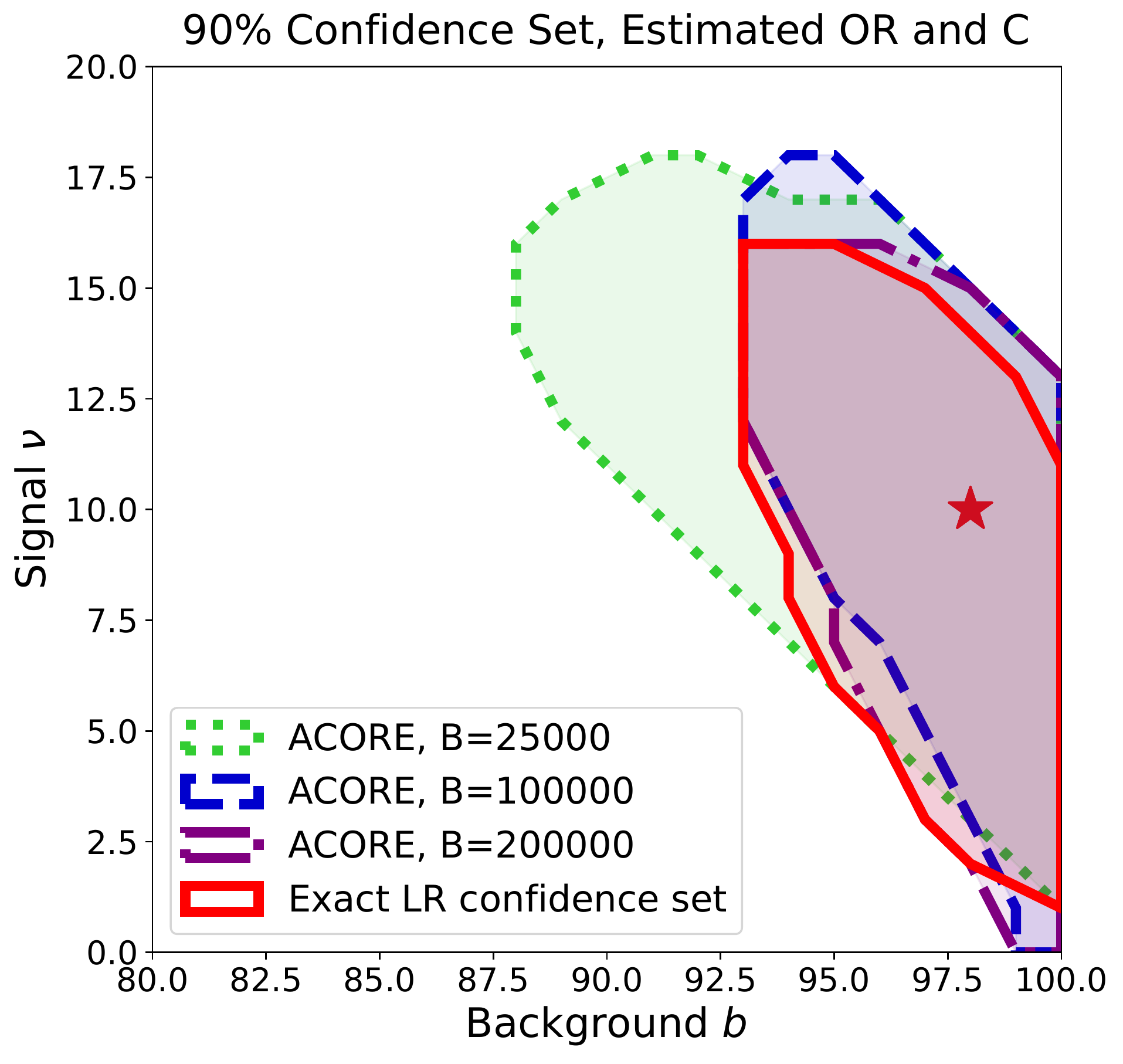} 
    \includegraphics[width=0.33\textwidth]{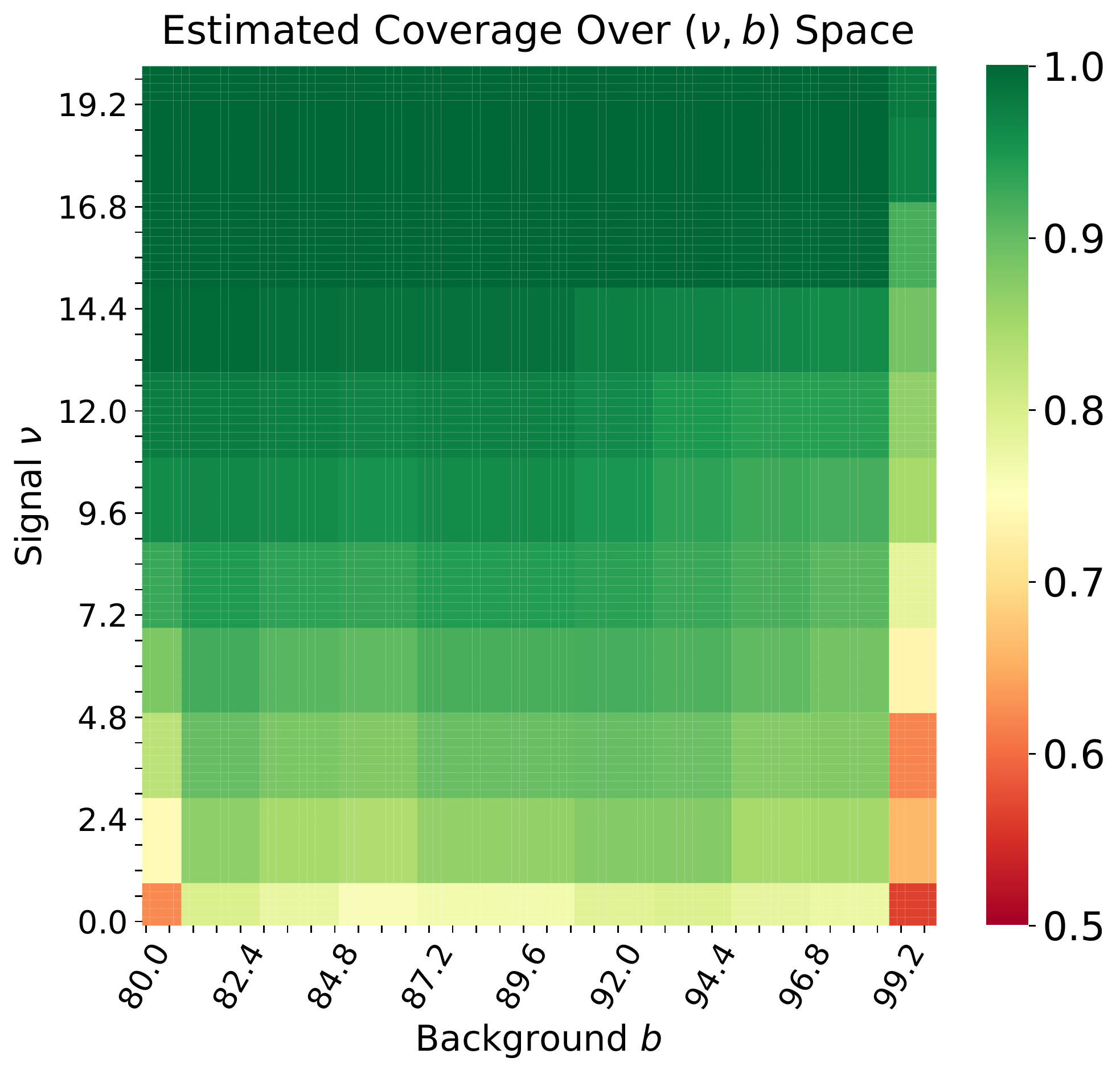} 
    \vspace*{-6mm}
    \caption{
    Signal detection HEP example.
    \textit{Left}: $90\%$ confidence sets computed with the exact likelihood ratio statistic. Estimating critical values can however be challenging, as highlighted by the differences in the results for two different quantile regression (QR) algorithms and sample sizes: Random Forest QR at $B'=1000$ (green dotted) versus Deep QR at $B'=25000$ (blue dashed). Our goodness-of-fit procedure can be used to select the best method in a principled way. (The red contour shows the exact LR confidence set, and the red star is at the true parameter setting.)
    \textit{Center}: $90\%$ confidence sets when using \texttt{ACORE} to estimate both odds ratios and critical values. This is the LFI setting.
    Our proposed strategy for choosing \texttt{ACORE} components 
    selects a 5-layer deep neural network with $B=100000$; this yields a confidence set (dashed blue) close to the exact LR set (solid red).
    Increasing $B$ does not show a noticeable improvement (dash-dotted purple), whereas decreasing $B$ makes estimates worse (dotted green).
    \textit{Right}: Heat map of the estimated coverage for a confidence set that did \emph{not} pass our goodness-of-fit diagnostic. The overall coverage of the confidence set is correct ($91.8\%$ vs. the $90\%$ nominal confidence level), but the set clearly undercovers in low-signal and high-background regions.
    } 
    \label{fig:hep_application_figure}
\end{figure*}

Next we illustrate our goodness-of-fit procedure (Section~\ref{sec:coverage_gof}) for checking the coverage of the constructed confidence sets  across the parameter space $\Theta$. To pass our goodness-of-fit test, we require the nominal coverage to be within two standard deviations of the estimated coverage for all parameter values. Figure ~\ref{fig:toy_example_coverage} shows the estimated coverage with logistic regression for the Poisson example with $B=1000$ (the training sample size for estimating odds via QDA) and three different values of $B'$ (the training sample size for estimating the critical value $C$ via gradient boosted quantile regression). As expected  (Theorem~\ref{thm:convergenceCutoffs}), the estimated coverage gets closer to the nominal $90\%$ confidence level as $B'$ increases. We can use these diagnostic plots to choose $B'$. For instance, here $B'=500$ is large enough for  \texttt{ACORE}  to achieve good coverage. (See Supplementary Material~\ref{sec:toy_examples_supp_mat} for a detailed analysis of this example.)

Supplementary Materials~\ref{supp_mat: comp_gp_mc} and \ref{supp_mat: comp_carl} include a comparison between \texttt{ACORE} and Monte Carlo Gaussian Process (MC GP) interpolation \cite{Frate2017GP_LFI} and calibrated neural nets classifiers (CARL, \citealt{Cranmer2015LikRatio}), respectively. Our results show that MC-based GP interpolation provides a better approximation of the likelihood ratio when the simulated data are approximately Gaussian (as in the Poisson example). However, when the parametric assumptions are not valid (as in the GMM example), MC-based GP fails to approximate the likelihood ratio regardless of the number of available simulations. For both examples, CARL leads to lower power and larger confidence intervals than ACORE. See Tables~\ref{tab:comparison_with_MC} and \ref{tab:comparison_with_carl} for details.

\section{Signal Detection in High Energy Physics}
\label{sec: applications}

In order to apply
\texttt{ACORE}, we need to choose four key components: (i) a probabilistic classifier, (ii) a training sample size $B$ for learning odds ratios, (iii) a
quantile regression algorithm, and (iv)
a training sample size $B^{'}$ for estimating critical values. 
We propose the following practical strategy to 
choose such components:

\begin{enumerate}
    \item Use the cross entropy loss to select the classifier and
    $B$
    (as seen in Section~\ref{sec: toy_example}, 
    a small cross entropy corresponds to higher power and a smaller confidence set);
    \vspace{-2mm}
    \item Then use our goodness-of-fit procedure (Section~\ref{sec:coverage_gof}) to select the quantile regression method and
    $B^{'}$.
\end{enumerate}

We illustrate \texttt{ACORE} and this strategy on a model described in 
\citet{Rolke2005PoissonCount} and \citet{Sen2009NuisanceParameters} 
for a high energy physics (HEP) experiment. In this model, particle collision events are counted under the presence of a background process $b$. The goal is to assess the intensity $\nu$ of a signal (i.e., an event which is not part of the background process).
The observed data $D$ consist of $n=10$ realizations of $\mathbf{X}=(N, M)$, where $N \sim \Pois(b + \nu)$ is the number of events in the signal region, and $M \sim \Pois(b)$ is  the number of events in the background (control) region.
(We use a uniform proposal distribution $r_{\Theta}$ and a Gaussian reference distribution $G$.) This model is a simplified version of a real particle physics experiment where the true likelihood function is not known.

Figure~\ref{fig:hep_application_figure} 
illustrates the role of $B$, $B'$,
and our goodness-of-fit procedure when estimating confidence sets. (For details, see Supplementary Material~\ref{appendix:hep_application}.)
In the {\em left} panel, we use the true LR statistic to show that, even if the LR is available, estimating the critical value $C$ well still matters. Our goodness-of-fit diagnostic provides a principled way of choosing the best quantile regression (QR) method and the best sample size $B'$ for estimating $C$.
In this example, random forest QR does not pass our goodness-of-fit test; it also leads to a confidence region quite different from the exact one. Deep QR, which passes our test, gives a more accurate region estimate. 
In the {\em center} panel, we use 
$\texttt{ACORE}$ to estimate both the odds ratio and the critical value $C$ (this is the LFI setting). If we choose $B$ by identifying when the cross entropy loss levels off, we would choose $B=100000$. Decreasing $B$ leads to a worse cross-entropy loss and, as the figure shows, also a larger confidence region. Increasing $B$ beyond our selected sample size does not lead to substantial gains. 
The {\em right} panel illustrates how our goodness-of-fit procedure can be used to identify regions in parameter space where a constructed confidence set is not valid. The heat map refers to an example which did \emph{not} pass our goodness-of-fit procedure. While the overall (marginal) coverage is at the right value, our diagnostic procedure (for estimating coverage as a {\em function} of $\nu$ and $b$) is able to identify undercoverage in low-signal and high-background regions. That is, for a valid confidence set, one needs to better estimate the critical value $C$ by, e.g., using a different quantile regression estimator or by increasing $B'$ (either uniformly over the parameter space or by an active learning scheme which increases the number of simulations at parameter settings where one undercovers).

\section{Conclusions} 
In this paper we introduce \texttt{ACORE}, a framework for carrying out frequentist inference in LFI settings.
\texttt{ACORE} is well suited for settings
with costly simulations, as it efficiently estimates  test statistics and critical values across the entire parameter space. We provide a new goodness-of-fit procedure for estimating coverage of constructed confidence sets for all possible parameter settings. Even if the likelihood ratio is not well estimated,
\texttt{ACORE} provides valid inference
as long as hypothesis tests and confidence sets pass our goodness-of-fit procedure (albeit at the cost of having less power and larger sets).  We provide practical guidance on how to choose the smallest number of simulations to guarantee powerful and valid procedures. 


Future studies will investigate the effect of $G$
and $r_{\Theta}$ on performance, as well as how \texttt{ACORE} scales with increasing (a) feature space dimension and (b) parameter space dimension. Because we utilize ML methods to efficiently estimate odds ratio and critical values (Algorithms~\ref{alg:joint_y} and \ref{alg:estimate_thresholds2}), performance in (a) will depend on the convergence rates of the chosen probabilistic classifier and quantile regression method. For (b), scaling relies on having an efficient search algorithm; this search is challenging for all likelihood-based methods. Common solutions include gradient-free optimization methods, such as Nelder-Mead \cite{Nelder1965NelderMead} and Bayesian optimization \cite{snoek2012BayesOpt}, and approximation techniques, such as profile likelihoods \cite{Murphy2000ProfileLikelihood} and hybrid resampling \cite{chaung2000hybridresampling, Sen2009NuisanceParameters}. Such approaches can potentially be integrated into \texttt{ACORE}. The \texttt{ACORE} framework can also be adapted to accommodate test statistics such as the Bayes factor \cite{Kass1995BayesFactors}. In addition to Bayes factors, we will investigate choosing the 
number of simulations $B$ via sequential testing and likelihood goodness-of-fit tests such as \citet{pmlr-v108-dalmasso20a}. In addition, we will consider extending the \texttt{ACORE} framework to include other statistical quantities in the likelihood ratio estimation process by, for example, adapting the regression on likelihood ratio (ROLR) score in \citet{Brehmer:2018hga}.
Finally, we will
include a theoretical study of how the power of the \texttt{ACORE} test relates to classifier performance.

\section*{Acknowledgments}

We thank the anonymous reviewers for their thoughtful comments and suggestions. ND is grateful to Tudor Manole, Alan Mishler, Aleksandr Podkopaev and the STAMPS research group for insightful discussions. RI is grateful for the financial support of FAPESP (2019/11321-9) and CNPq (306943/2017-4). The early stages of this research were supported in part by the National Science Foundation under DMS-1520786.

\clearpage 

\appendix

\begin{center}
\textbf{\Large \bf Supplementary Material: Confidence Sets and Hypothesis Testing in a Likelihood-Free Inference Setting}
\end{center}

\section{Algorithm for Simulating Labeled Sample for Estimating Odds Ratios}
\label{appendix::simulation_training}

Algorithm~\ref{alg:joint_y} provides details on how to create the training sample $\T_B$ for estimating odds ratios. Out of the total training sample size $B$, a proportion $p$ is generated by the stochastic forward simulator $F_\theta$ at different parameter values $\theta$, while the remainder is sampled from a reference distribution $G$.

\begin{algorithm}[!ht]
	\caption{Generate a labeled sample of size $B$ for estimating odds ratios
	 \small }\label{alg:joint_y}
	\algorithmicrequire \ {\small   stochastic forward simulator $F_\theta$; reference distribution $G$, proposal distribution $r_{\Theta}$ over parameter space; 
	training sample size B; parameter p of Bernoulli distribution}\\
	\algorithmicensure \ {\small labeled training sample}
	\begin{algorithmic}[1]
		\STATE Draw parameter values $\theta_1, \ldots, \theta_B \stackrel{iid}{\sim} r_{\Theta}$
		\STATE Assign labels $Y_1, \ldots, Y_B \stackrel{iid}{\sim}{\sim Ber}(p)$
		\FOR{$i=1,\ldots,B$}
		\IF{$Y_i==1$}
		 \STATE  Draw sample from forward simulator, $\X_i \sim F_\theta$
		 \ENDIF
	\IF{$Y_i==0$}
		\STATE Draw sample from reference distribution, $\X_i \sim G$
		 \ENDIF
	\ENDFOR
		\STATE \textbf{return} $\T_B=\{\theta_i,\X_i,Y_i\}_{i=1}^B$  
	\end{algorithmic}
\end{algorithm}

\section{Algorithm for Constructing Confidence Set for $\theta$}\label{appendix::confidence_set}

Algorithm~\ref{alg:conf_reg}  summarizes the \texttt{ACORE} procedure for constructing confidence sets. First, we estimate parametrized odds to compute the \texttt{ACORE} test statistic $\tau$ (Eq.~\ref{eq:odds_ratio_statistic}). Then, we compute a parametrized estimate of the critical values as  a conditional distribution function of $\tau$. Finally, we compute a confidence set for $\theta$ by the Neyman inversion technique \citep{neyman1937inversion}.

 \begin{algorithm}[!ht]
	\caption{
	Construct confidence set for $\theta$ with  coefficient $\gamma=1-\alpha$
	\small}\label{alg:conf_reg}
	\algorithmicrequire \ {\small  stochastic forward simulator $F_\theta$;  reference distribution $G$; proposal distribution r over $\Theta$;  parameter p of Bernoulli distribution; sample size $B$ (for estimating odds ratios); sample size $B'$ (for estimating critical values); probabilistic classifier; observed data
	$D=\left\{\x_1^{\text{obs}},\ldots,\x_n^{\text{obs}}\right\}$; $\alpha \in (0,1)$
	}\\
	\algorithmicensure \ {\small $\theta$-values in confidence set}
	\begin{algorithmic}[1]
\STATE \codecomment{Estimate odds ratios:}
		\STATE Generate labeled sample $\mathcal{T}_B$ according to Algorithm~\ref{alg:joint_y}
		\STATE Apply probabilistic classifier to $\mathcal{T}_B$ to learn class posterior probabilities,  $\widehat{\P}(Y=1|\theta,\X),$ for all $\theta \in \Theta$ and $\X \in \mathcal{X}$	
    \STATE Let the estimated odds $\widehat{\O}(\theta,\X) = \frac{\widehat{\P}(Y=1|\theta,\X)}{\widehat{\P}(Y=0|\theta,\X)}$
        \STATE Let the estimated odds ratios 
    $\widehat{\Or}(\X;\theta_0,\theta_1) = \frac{\widehat{\O}(\theta_0,\X)}{\widehat{\O}(\theta_1,\X)},$ 
    for all $\theta_1, \theta_2 \in \Theta$ and $\X \in \mathcal{X}$
\STATE  \codecomment{Estimate the critical value for a test $\delta_{\theta_0}$ that rejects $\theta=\theta_0$ at significance level $\alpha$:}
    \STATE Construct parametrized function  $\widehat{C}_{\theta_0} := \widehat{F}_{\tau|\theta_0}^{-1}(\alpha | \theta_0)$
     for $\theta_0 \in \Theta$ according to Algorithm~\ref{alg:estimate_thresholds_conf2}
\STATE  \codecomment{Find parameter set for which the test  $\delta_{\theta_0}$ does not reject $\theta=\theta_0$:}
    \STATE \texttt{ThetaGrid} $\leftarrow$ grid of parameter values in $\Theta$
    \STATE $n_{\text{grid}} \leftarrow $ length(\texttt{ThetaGrid})
    \STATE Set $S \gets \emptyset$
    \FOR{$\textrm{\texttt{Theta}0} \in  \texttt{ThetaGrid}$} 
 	      		\STATE  $ \textrm{\texttt{Cutoff}} \gets \widehat{F}_{\tau | \theta_0}^{-1}\left(\alpha \ \middle\vert \  \textrm{\texttt{Theta}0}\right)$ 
		\STATE $\textrm{\texttt{sumLogOR}} \gets$ array with length  $n_{\text{grid}}$
         	\FOR{$j=1, \ldots, n_{\text{grid}}$} 
       \STATE $ \textrm{\texttt{Theta}1} \gets \textrm{\texttt{ThetaGrid}}[j]$ 
        \STATE  $ \textrm{\texttt{sumLogOR}}[j] \gets  \sum_{k=1}^n  \log \left( \widehat{\Or}(\x_k^{\text{obs}};  \textrm{\texttt{Theta}0},  \textrm{\texttt{Theta}1})\right)$ 
        \ENDFOR
      \STATE  $ \textrm{\texttt{TauObs}} \gets \min (  \textrm{\texttt{sumLogOR}}))$ 
              \IF{ $ \textrm{\texttt{TauObs} }>  \textrm{\texttt{Cutoff}}$ }
             \STATE  $S \gets  S  \cup   \textrm{\texttt{Theta}}0$ 
             \ENDIF
  \ENDFOR
		\STATE \textbf{return} S
	\end{algorithmic}
\end{algorithm}

\section{Proofs} \label{appendix:proofs}

\begin{proof}[Proof of Proposition \ref{prop::consistency}]
If $\ \widehat{\P}(Y=1|\theta,\x)=\P(Y=1|\theta,\x)$, then $\widehat{\Or}(\x;\theta_0,\theta_1) = \Or(\x;\theta_0,\theta_1)$. By Bayes rule and construction
(Algorithm \ref{alg:joint_y}), 
$${\O}(\x;\theta):=\frac{\P(Y=1|\theta,\x)}{\P(Y=0|\theta,\x)}=\frac{f(\x|\theta)p}{g(\x)(1-p)}.$$
 Thus,  the odds ratio at $\theta_0,\theta_1 \in \Theta$ is given by
$$ \Or(\x;\theta_0,\theta_1)= \frac{f(\x|\theta_0)}{f(\x|\theta_1)},$$ 
and therefore 
\begin{align*}
\tau(D;\Theta_0)=&\sup_{\theta_0 \in \Theta_0} \ \inf_{\theta_1 \in \Theta} \sum_{i=1}^n \left(  \log \widehat{\Or}(\x_i^{\text{obs}};\theta_0,\theta_1)\right)\\
&=\sup_{\theta_0 \in \Theta_0} \ \inf_{\theta_1 \in \Theta} \sum_{i=1}^n  \log \frac{f(\x_i^{\text{obs}}|\theta_0)}{f(\x_i^{\text{obs}}|\theta_1)} 
\\
&=\sup_{\theta_0 \in \Theta_0} \ \inf_{\theta_1 \in \Theta}
\log \left( \frac{\L(D;\theta_0)}{\L(D;\theta_1)}\right)\\
&=\Lambda(D;\Theta_0).
\end{align*}
\end{proof}

 \begin{proof}[Proof of Theorem~\ref{thm:convergenceCutoffs}]
 
 The union bound and Assumption~\ref{assum:quantile_consistent} imply that
 $$\sup_{\theta \in \Theta_0} \sup_{t \in \mathbb{R}}|\hat F_{B'}(t|\theta)- F(t|\theta)|\xrightarrow[B' \longrightarrow\infty]{\enskip \P \enskip} 0.$$
 It follows that
 $$\sup_{\theta \in \Theta_0} |\hat F^{-1}_{B'}(\alpha|\theta)- F^{-1}(\alpha|\theta)|\xrightarrow[B' \longrightarrow\infty]{\enskip \P \enskip} 0.$$
 The result follows from the fact that
 \begin{align*}
 0 \leq    |C_{B,B'}-C_B^*|&= |\sup_{\theta \in \Theta_0} \hat F^{-1}_{B'}(\alpha|\theta)-\sup_{\theta \in \Theta_0}  F^{-1}(\alpha|\theta)| \\
     &\leq \sup_{\theta \in \Theta_0} |\hat F^{-1}_{B'}(\alpha|\theta)- F^{-1}(\alpha|\theta)|,
 \end{align*}
 and thus 
 $$|C_{B,B'}-C_B^*| \xrightarrow[B' \longrightarrow\infty]{\enskip \P \enskip} 0.$$
 
 \end{proof}

\begin{lemma}
 \label{lemma:convergence_statistic}
 If
 $(\widehat{\P}(Y=1|\theta,\X))_{\theta \in \Theta} \xrightarrow[B \longrightarrow\infty]{\enskip \P \enskip} (\P(Y=1|\theta,\X))_{\theta \in \Theta}$
 and $|\Theta|<\infty$, then
 $$\tau(\D; \Theta_0) \xrightarrow[B \longrightarrow\infty]{\enskip \P \enskip}
 \sup_{\theta_0 \in \Theta_0}\inf_{\theta_1 \in \Theta} \sum_{i=1}^n  \log \left( \Or(\X_i^{\text{obs}};\theta_0,\theta_1)\right)$$
 \end{lemma}
 
  \begin{proof}
 For every $\theta_0,\theta_1 \in \Theta$,
 it follows directly from the properties of convergence in probability that
 \begin{align*}
    \sum_{i=1}^n  \log & \left( \widehat{\Or}(\X_i^{\text{obs}};\theta_0,\theta_1)\right) \\ & \xrightarrow[B \longrightarrow\infty]{\enskip \P \enskip}
 \sum_{i=1}^n  \log \left( \Or(\X_i^{\text{obs}};\theta_0,\theta_1)\right) 
 \end{align*}

 The conclusion of the lemma follows from the continuous mapping theorem.
 \end{proof}

 \begin{proof}[Proof of Theorem~\ref{thm:convergenceLRT}]
 Lemma~\ref{lemma:convergence_statistic}
 implies that 
 $\tau_B(\D; \Theta_0)$ converges in distribution to 
 $\sup_{\theta_0 \in \Theta_0}\inf_{\theta_1 \in \Theta}\sum_{i=1}^n  \log \left( \Or(\X_i^{\text{obs}};\theta_0,\theta_1)\right)$. 
 Now, from Slutsky's theorem,
 \begin{align*}
    \tau_B&(\D; \Theta_0)-\widehat{C}_B \\
 &\xrightarrow[B \longrightarrow\infty]{\enskip \mbox{Dist} \enskip}\sup_{\theta_0 \in \Theta_0}\inf_{\theta_1 \in \Theta}
 \sum_{i=1}^n  \log \left( \Or(\X_i^{\text{obs}};\theta_0,\theta_1)\right)-C^*. 
 \end{align*} 
It follows that
 \begin{align*}
     &\P\left(\widehat{\phi}_{B,\widehat{C}_B}(\mathcal{D})=1|\theta \right)=
     \P\left(\tau_B(\D; \Theta_0)-\widehat{C}_B\leq 0|\theta \right) \\
     &\xrightarrow[B \longrightarrow\infty]{}
     \P\Big(\sup_{\theta_0 \in \Theta_0}\inf_{\theta_1 \in \Theta} \sum_{i=1}^n  \log \left( \Or(\X_i^{\text{obs}};\theta_0,\theta_1)\right)\\
     &\hspace{3.5cm}  -C^* \leq 0|\theta \Big)\\
     &=\P\left(\phi_{C^*}(\mathcal{D})=1|\theta\right),
 \end{align*}
 where the last equality follows from Proposition~\ref{prop::consistency}.
 
 \end{proof}
 
\section{Toy Examples}\label{sec:toy_examples_supp_mat}

This section provides details on the toy examples of Section~\ref{sec: toy_example}.
We use the \texttt{sklearn} ecosystem \cite{scikit-learn} implementation of the following probabilistic classifiers:
\vspace{-3mm}
\begin{itemize}
    \item multi-layer perceptron (MLP) with default parameters, but no $L^2$ regularization ($\alpha =0$);
    \vspace{-2mm}
    \item quadratic discriminant analysis (QDA) with default parameters;
    \vspace{-2mm}
    \item nearest neighbors (NN) classifier, with number of neighbors equal to the rounded square root of the number of data points available (as per \citet{Duda01PatternClassification}).
\end{itemize}
\vspace{-3mm}

Table~\ref{tab:supp_mat_toy_example_coverage} reports the observed coverage for the settings of Tables~\ref{tab:toy_example_setup} and \ref{tab:cross_ent_poisson}. Critical values or $C$ are estimated with quantile gradient boosted trees ($100$ trees with maximum depth equal to $3$), a training sample size $B^{'}=5000$,  observed data $D$ of sample size $n=10$, nominal coverage of $90\%$, and averaging over $100$ repetitions. The table shows that we for all cases achieve results in line with the nominal 
confidence level.\footnote{The $95\%$ CI of a binomial distribution with probability $p=0.9$ over $100$ repetitions is in fact $[0.84, 0.95]$. This interval includes the observed coverages listed in Table~\ref{tab:supp_mat_toy_example_coverage}.} 

\begin{table}[!ht]
\centering
\begin{tabular}{|c|c|c|c|}
\hline
 $B$ & \textit{Classifier} & \textit{Poisson Example}  &\textit{GMM Example} \\ 
  & & \textit{Coverage}  &\textit{Coverage} \\ \hline
\multirow{3}{*}{100} & MLP & 0.91 & 0.87\\
 & NN & 0.91 & 0.91\\
 & QDA & 0.90 & 0.88\\ \hline
\multirow{3}{*}{500} & MLP & 0.91 & 0.91\\
 & NN & 0.93 & 0.95\\
 & QDA & 0.94 & 0.92\\ \hline
\multirow{3}{*}{1000} & MLP & 0.91 & 0.92\\
 & NN & 0.89 & 0.88\\
 & QDA & 0.91 & 0.93\\ \hline
\end{tabular}
\caption{Observed coverage of the toy examples in Tables ~\ref{tab:toy_example_setup} and \ref{tab:cross_ent_poisson}. These values are consistent with what we would expect for 100 trials with a nominal confidence level of  $90\%$; see text.}\label{tab:supp_mat_toy_example_coverage}
\end{table}

Our goodness-of-fit procedure shown in Figure~\ref{fig:toy_example_coverage} uses a set $\mathcal{T}^{''}_{B^{''}}$ with size $B^{''}=250$ (as defined in Section~\ref{sec:coverage_gof}); Figure~\ref{fig:supp_mat_est_coverage_gmm} shows the goodness-of-fit plot for the Gaussian mixture model example, where the coverage is estimated via logistic regression and the critical values are estimated via quantile gradient boosted trees. For the Poisson example a training sample size of $B'=500$ seems to be enough to achieve correct coverage, whereas the Gaussian mixture model example requires $B'=1000$. 

Next we compare  our goodness-of-fit diagnostic with diagnostics obtained via standard Monte Carlo sampling. Figure~\ref{fig:supp_mat_obs_coverage} shows the MC coverage as a function of $\theta$ for the Poisson example (left) and the Gaussian mixture model example (right). In both cases 100 MC samples are drawn at 100 parameter values chosen uniformly. The empirical \texttt{ACORE} coverage is computed over the MC samples at each chosen $\theta$. This MC procedure is expensive: it uses a total of $10,000$ simulations, which is $40$ times the number used in our goodness-of-fit procedure. The observed coverage of the Poisson example (Figure~\ref{fig:supp_mat_obs_coverage}, left)  indicates that $B'=500$ is sufficient to achieve the nominal coverage of $90\%$. For the Gaussian mixture model example (Figure~\ref{fig:supp_mat_obs_coverage}, right), we detect undercoverage for very small values of $\theta$. This discrepancy is due to the fact that, at $\theta=0$, the mixture collapses into a single Gaussian, structurally different from the GMM at any other $\theta > 0$ and closer to the $\mathcal{N}(0, 5^2)$ reference distribution. 

Our goodness-of-fit procedure is able to identify that the actual coverage is far from the nominal coverage at small values of $\theta$, when the training sample size $B'$ for estimating $C$ is too small. More specifically, Figure~\ref{fig:supp_mat_est_coverage_gmm} shows a noticeable tilt in the prediction bands for $B'=100$ and $500$. However, as $B'$ increases, the estimation of critical values becomes more precise and the estimated confidence intervals pass our goodness-of-fit diagnostic at, for example, $B'=1000$. Future studies will provide a more detailed account on how such boundary effects depend on the method for estimating the coverage.

\begin{figure}[!ht]
    \centering
    \includegraphics[width=0.495\textwidth]{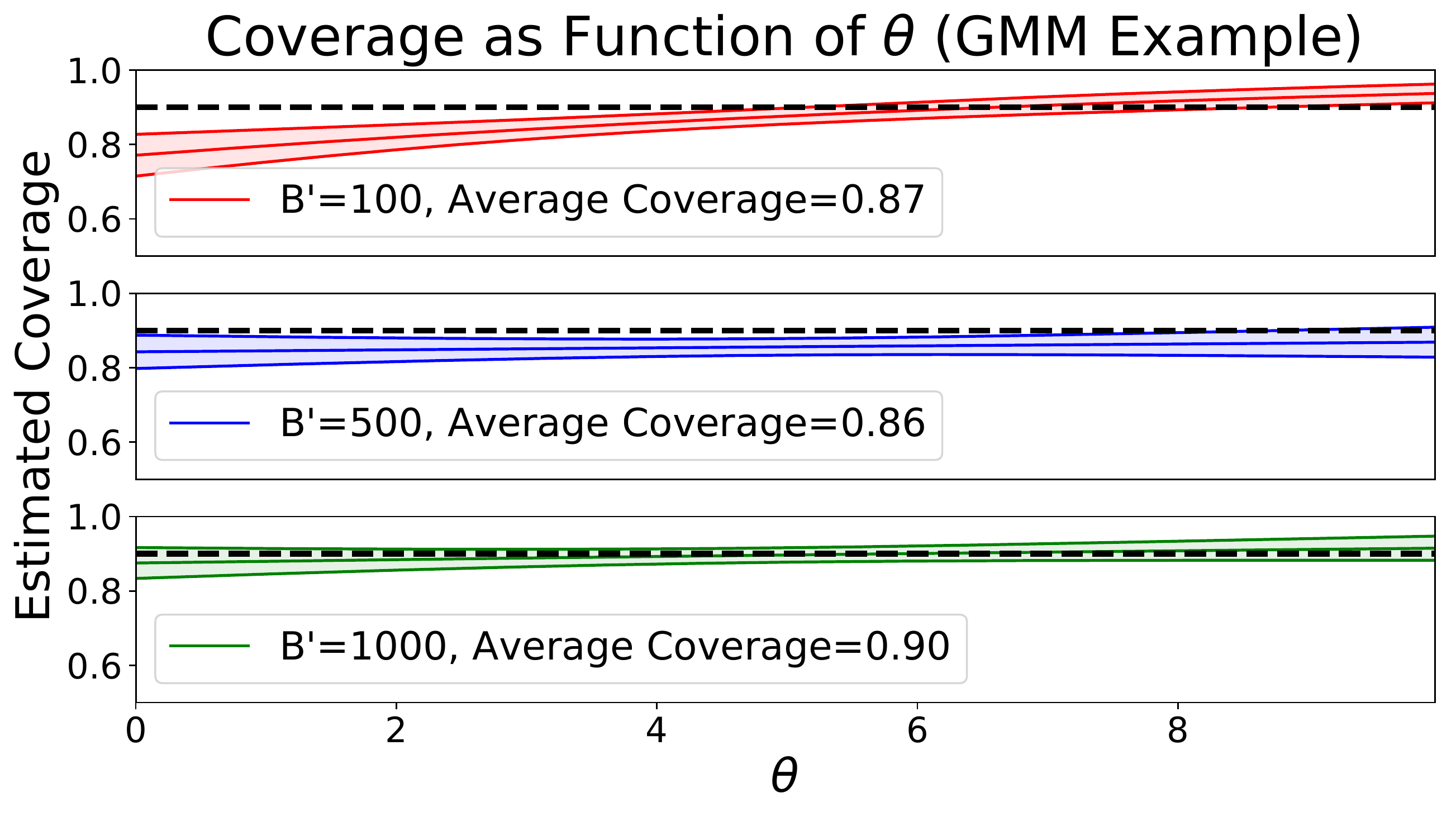}
    \vspace*{-9mm}
    \caption{Estimated coverage as a function of $\theta$ in the Gaussian mixture model example for \texttt{ACORE} with different values of $B'$. Logistic regression is used to estimate mean coverage and one standard deviation prediction bands. (We here use $n=10$, a MLP 
    classifier with $B=1000$ and quantile gradient boosted trees).}
    \label{fig:supp_mat_est_coverage_gmm}
\end{figure}

\begin{figure*}[!ht]
    \centering
    \includegraphics[width=0.495\textwidth]{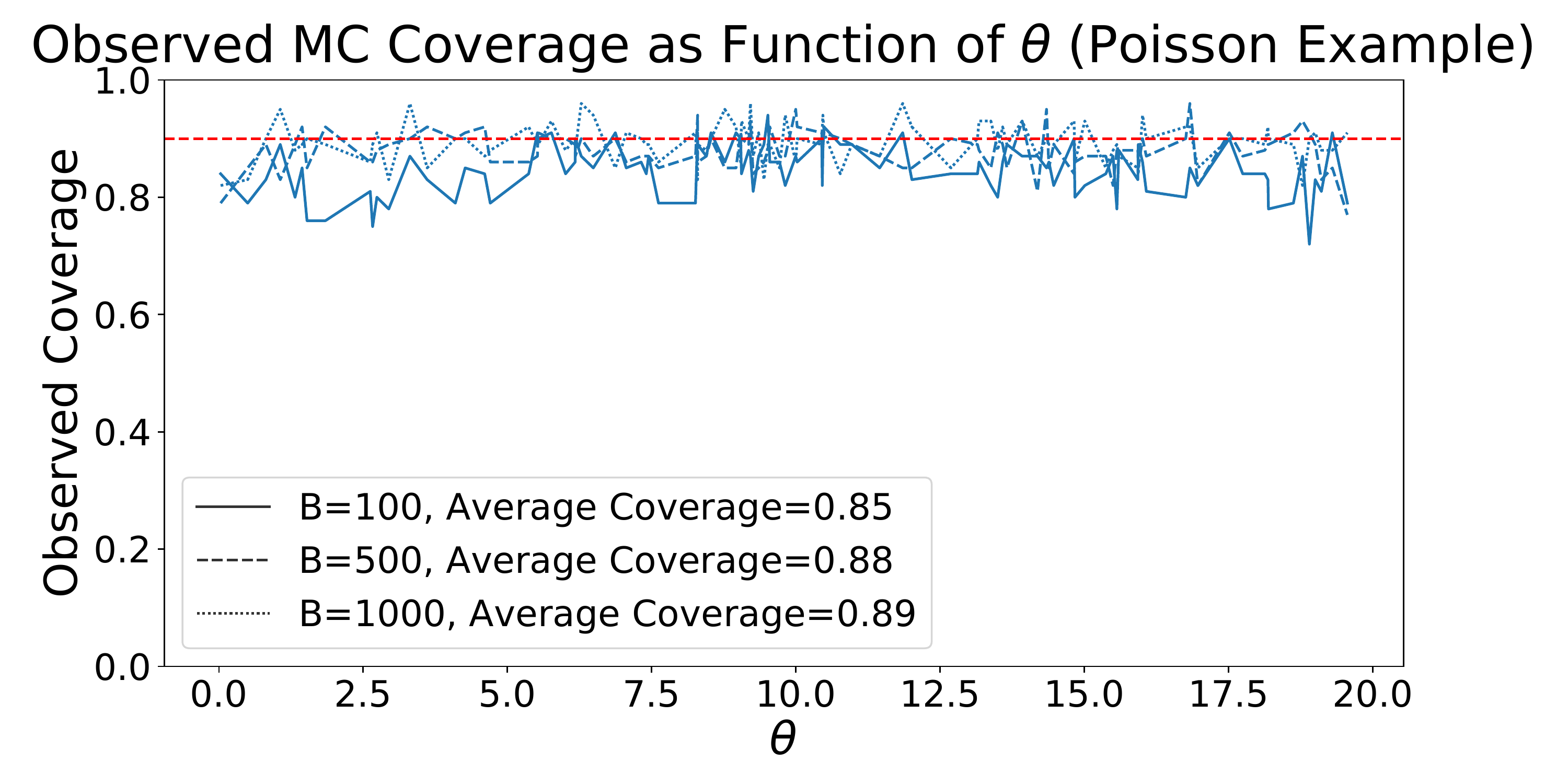}
    \includegraphics[width=0.495\textwidth]{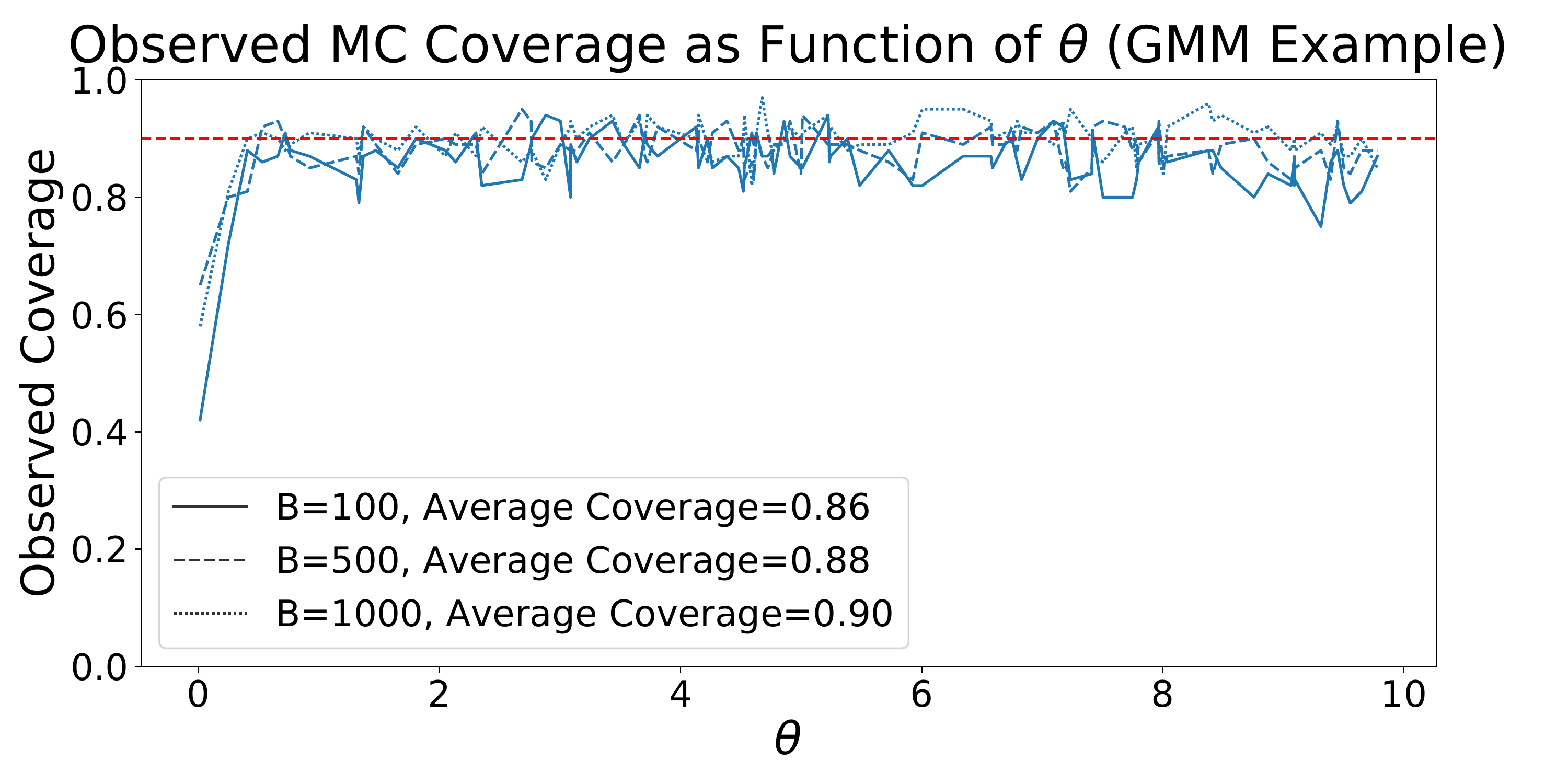}
    \vspace*{-5mm}
    \caption{Observed \texttt{ACORE} coverage across the parameter space for the Poisson example (left) and the Gaussian mixture model example (right). The coverage is computed with Monte Carlo samples of size $100$, each sampled at a $\theta$ chosen uniformly over the parameter space. Odds ratios are computed with a QDA classifier for the Poisson example, and an MLP classifier for the GMM example (as in Figures~\ref{fig:toy_example_coverage} and \ref{fig:supp_mat_est_coverage_gmm}). We observe undercoverage at small $\theta$ for the GMM (right) due to the mixture collapsing into a single Gaussian as $\theta \to 0$.}
    \label{fig:supp_mat_obs_coverage}
\end{figure*}

\section{Signal Detection in High Energy Physics}\label{appendix:hep_application}

Here we consider the signal detection example in Section~\ref{sec: applications} We describe the details of the construction of \texttt{ACORE} confidence sets which
used the strategy in Section~\ref{sec: applications} to choose \texttt{ACORE} components and parameters. For learning the odds ratio, we compared the following classifiers:
\vspace{-3mm}
\begin{itemize}
    \item logistic regression,
    \vspace{-2mm}
    \item quadratic discimininant analysis (QDA) classifier,
    \vspace{-2mm}
    \item nearest neighbor classifier,
    \vspace{-2mm}
    \item gradient boosted trees using $\{100, 500, 1000\}$ trees with maximum depth $\{3,5,10\}$,
    \vspace{-2mm}
    \item Gaussian process classifiers\footnote{GP classifiers were used only with sample sizes $B$ below $10,000$, as the matrix inversion quickly becomes computationally infeasible for larger values of $B$.} with radial basis functions kernels with variance $\{1,.5,.1\}$,
    \vspace{-2mm}
    \item feed-forward deep neural networks, with ${2,...,6}$ deep layers, number of neurons between $2^{\{4,...,10\}}$ and either ReLu or hyperbolic tangent activations.
\end{itemize}

For  estimating the critical values, we considered the following quantile regression algorithms:
\vspace{-2mm}
\begin{itemize}
    \item gradient boosted trees using $\{100, 250, 500\}$ trees with maximum depth $\{3,5,10\}$,
    \vspace{-2mm}
    \item random forest quantile regression with $\{100, 250, 500\}$ trees,
    \vspace{-2mm}
    \item deep quantile regression with $\{2,3\}$ deep layers, $2^{\{4,..,6\}}$ neurons and ReLu activations (using the \texttt{PyTorch} implementation \cite{pytorch2019}).
\end{itemize}

All computations were performed on 8-Core Intel Xeon CPUs X5680 at 3.33GHz. 

Figure~\ref{fig:supp_mat_hep_example} illustrates the two steps in identifying the four components of \texttt{ACORE}.  We first  use a validation set of $5,000$ simulations to determine which probabilistic classifier and training sample size $B$ minimize the cross entropy loss. Figure~\ref{fig:supp_mat_hep_example} (left) shows the cross entropy loss of the best four classifiers as function of $B$. The minimum is achieved by a 5-layer deep neural network (DNN) at $B=100,000$ with a cross entropy loss of $58.509\times 10^{-2}$, closely followed by QDA with $58.512\times 10^{-2}$ at $B=50,000$. Given how similar the loss values are, we select both classifiers to follow-up on.
In Figure~\ref{fig:supp_mat_hep_example} (right), the ``estimated correct coverage'' represents the proportion of the parameter space that passes our diagnostic procedure.
The lowest $B'$ with correct coverage is achieved by the five-layer DNN classifier (for estimating odds ratios) at $B'=25,000$ with critical values estimated via a two-layer deep quantile regression algorithm. None of the quantile regression algorithms pass a diagnostic test with a nominal coverage of $90\%$ at the one standard deviation level when using the QDA classifier. We therefore do not use QDA in Section~\ref{sec: applications}.

Based on the analysis above, we choose the following \texttt{ACORE} components: (i) a five-layer DNN for learning odds ratios, (ii) $B=100,000$, (iii) a two-layer deep quantile regression for estimating critical values, and (iv) $B'=25,000$. Figure~\ref{fig:hep_application_figure} shows the confidence sets computed with this choice. 

\begin{figure*}[!ht]
    \centering
     \includegraphics[width=0.495\textwidth]{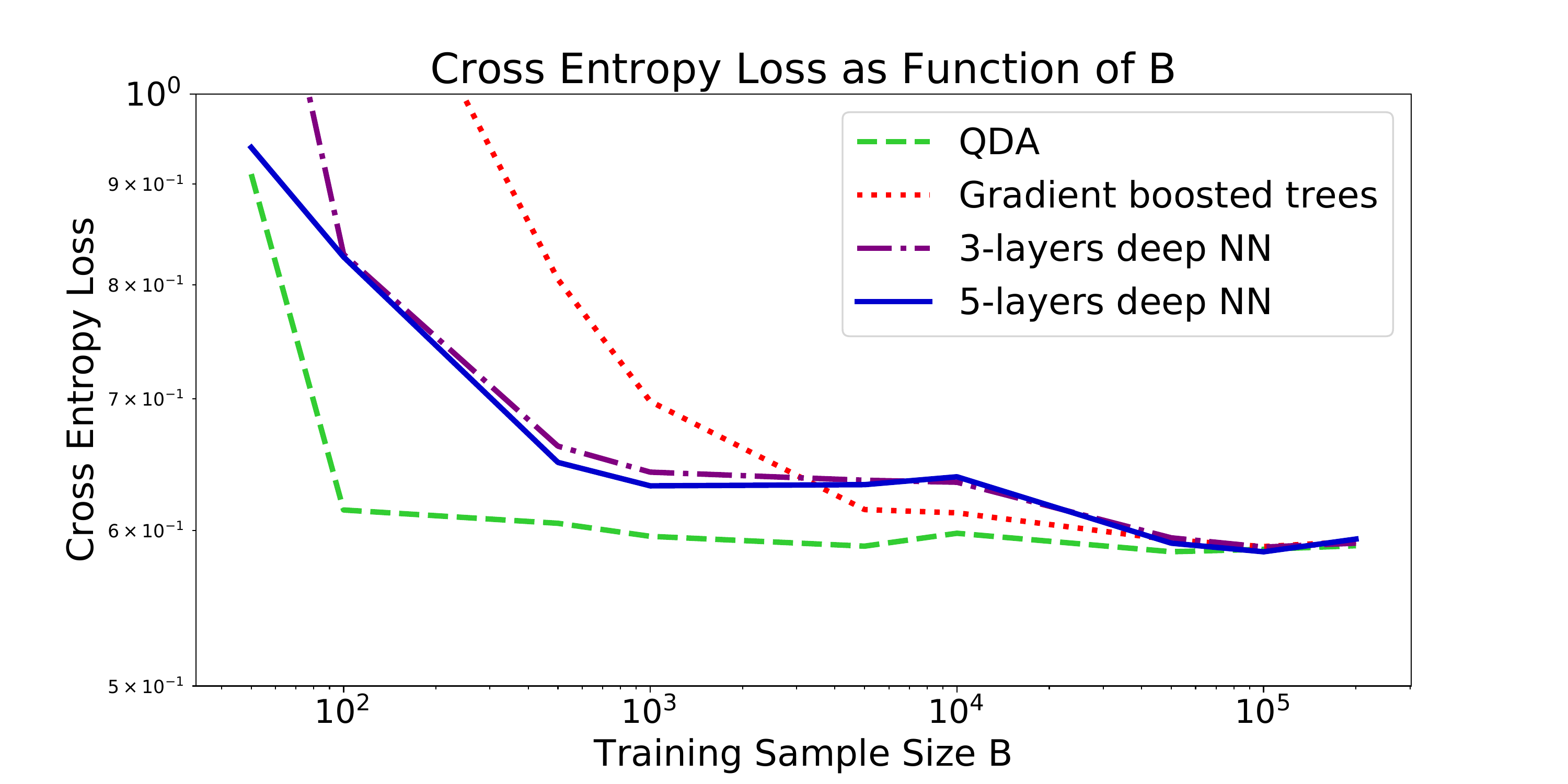}
    \includegraphics[width=0.495\textwidth]{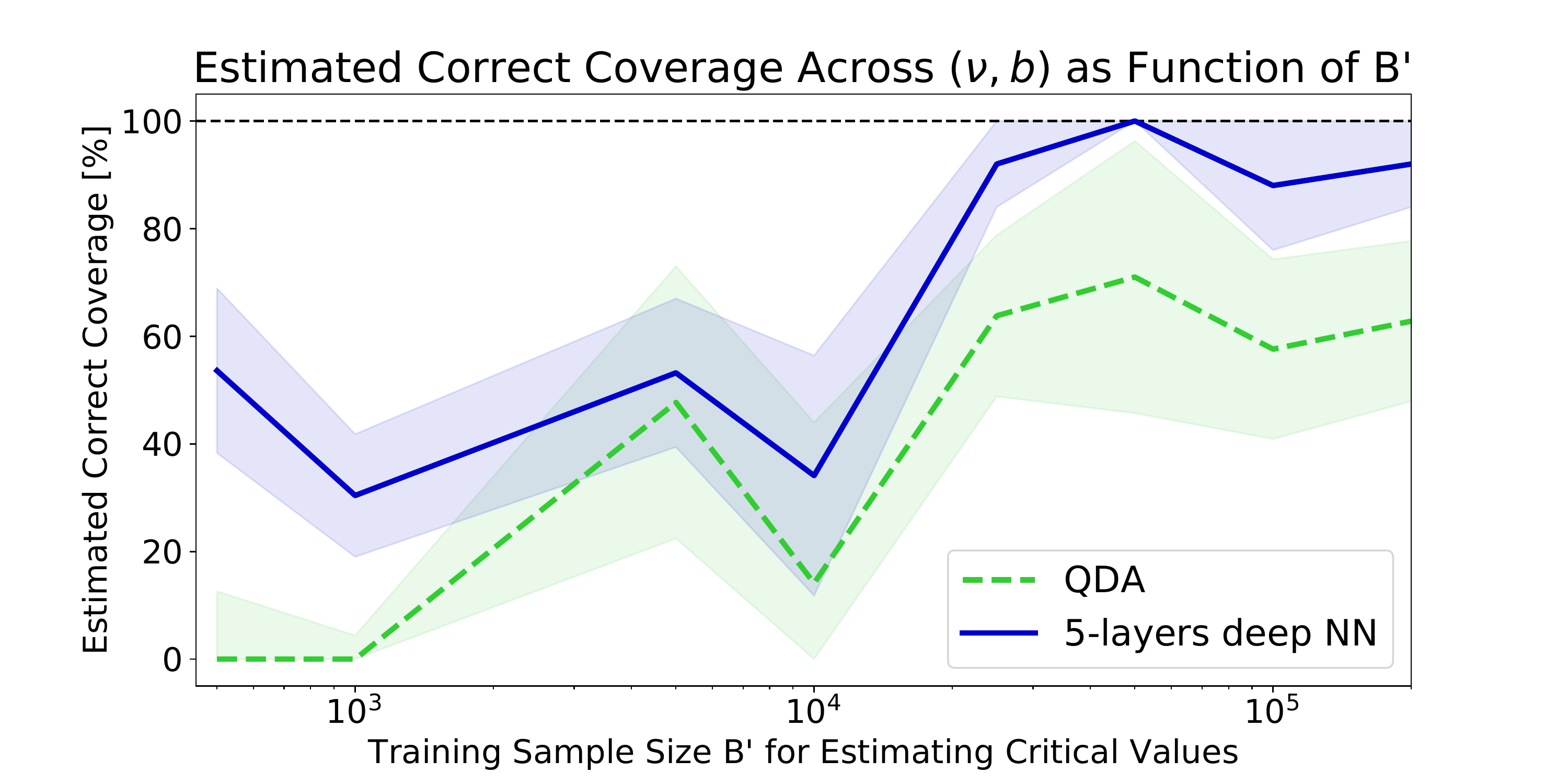}
    \caption{Using the strategy in Section~\ref{sec: applications} to choose \texttt{ACORE} components for the signal detection example. \textit{Left:} The cross entropy loss of the best four classifiers, shown as a function of $B$. In order of increasing loss: 5-layer DNN ([512, 256, 64, 32, 32] neurons, ReLu activations), QDA classifier, 3-layer DNN ([64, 32, 32] neurons, ReLu activations) and gradient boosted trees ($1000$ trees with maximum depth 5). Because the first two classifiers (the 5-layer DNN and QDA) achieve a very similar minimum loss, we consider both classifiers in the follow-up step.  \textit{Right}: Proportion of the $(\nu, b)$ parameter space where the best two classifiers pass our goodness-of-fit procedure with a nominal coverage of $90\%$. Both the mean value curves and the $\pm$ one standard deviation prediction bands are computed via logistic regression. Critical values are estimated via a two-layer deep quantile regression ([64,64] neurons, ReLu activations), which passed the diagnostic at the lowest sample size ($B'=25,000$, with the 5-layers DNN). Based on the results, we choose the 5-layer DNN with $B'=25,000$.}
    \label{fig:supp_mat_hep_example}
\end{figure*}

\section{Cross Entropy Loss Analysis}
\label{sec:cross_entropy_loss}

In this work, we use the cross entropy loss to measure the accuracy of the probabilistic predictions of the classifier.
That is, we calibrate the estimated odds function
$g(\theta,\x) := \widehat{\P}(Y=1 | \theta, \x)/\widehat{\P}(Y=0 | \theta, \x)$  
as follows:
Consider a sample point  $\{\theta,\x,y\}$ generated according to Algorithm \ref{alg:joint_y}.
Let
$p$ be a $\mbox{Ber}(y)$ distribution, and
$q$ be a  $\mbox{Ber}\left(\widehat{\P}(Y=1 | \theta, \x)\right)=\mbox{Ber}\left(\frac{g(\theta,\x)}{1+g(\theta,\x)}\right)$
distribution. The {\em cross entropy} between $p$ and $q$ is given by
\begin{align*} L_{\mbox{CE}}(g; \{\theta,\x,y\}) &= - y \log\left(\frac{g(\theta,\x)}{1+g(\theta,\x)}\right) \\
&\quad \quad \quad - (1-y) \log\left(\frac{1}{1+g(\theta,\x)}\right) \\
&= - y \log \left(g(\theta,\x) \right) + \log \left(1+g(\theta,\x) \right).
\end{align*}
For every $\x$ and $\theta$, the expected cross entropy  $\E[L_{\mbox{CE}}(g; \{\theta,\x,Y\})]$
is minimized by $g(\theta,\x)=\O(\theta,\x)$.
Thus we can measure the performance of an estimator $g$
of the odds
by the risk
$$ R_{\mbox{CE}}(g)=\E[L_{\mbox{CE}}(g; \{\theta,\X,Y\})] .$$
The cross entropy loss is not the only loss function that is minimized by the true odds function, but it is usually easy to compute in practice. It is also  well known that minimizing the cross entropy loss between the estimated distribution $q$ and the true distribution $p$ during training is equivalent to minimizing the Kullback-Leibler (KL) divergence between the two distributions, as 
$$KL(p ||  q) = H(p, q) - H(p),$$
where $H(p,q)$ is the cross entropy and $H(p)$ is the entropy of the true distribution. By Gibbs' inequality \citep{MacKay2002InformationTheory}, we have that $KL(p|| q) \geq 0$; hence the entropy $H(p)$ of the true distribution lower bounds the cross entropy with the minimum achieved when $p=q$. Hence, we can connect the cross entropy loss to the \texttt{ACORE} statistic.

\begin{Prop}\label{prop::cross_ent_lrt}
If the probabilistic classifier in \texttt{ACORE} achieves the minimum of the cross entropy loss, then the constructed \texttt{ACORE} statistic \eqref{eq:odds_ratio_statistic} is equal to the likelihood ratio statistic \eqref{eq::LRT}.
\end{Prop}
\begin{proof}[Proof of Prop~\ref{prop::cross_ent_lrt}]
The proof follows from  Proposition~\ref{prop::consistency} and the expected cross entropy loss is minimized if and only if $\widehat{\O}(\theta,\x)=\O(\theta,\x)$.
\end{proof}

In addition, we show that the convergence of the class posterior implies the convergence of the cross entropy to the entropy of the true distribution. This supports our decision to use the cross entropy loss when selecting the probabilistic classifier and sample size $B$.

\begin{lemma}\label{lemma::ce_int}
If for every $\theta \in \Theta$
$$q := \widehat{\P}(Y=1|\theta,\X) \xrightarrow[B \longrightarrow\infty]{\enskip \P \enskip} p:=\P(Y=1|\theta,\X),$$
then $H(p, q) \xrightarrow[B \longrightarrow\infty]{\enskip \P \enskip} H(p).$
\end{lemma}
\begin{proof}[Proof of Lemma~\ref{lemma::ce_int}]
We can rewrite the cross entropy $H(p, q)$ and entropy $H(p)$ as 
\begin{align*}
    H(p, q) & = - \sum_{y \in \{0,1\} } \int_{\mathcal{X} \times \Theta} p \log\left(q \right) {\rm d}\P(\x,\theta), \\
    H(p) & = - \sum_{y \in \{0,1\} } \int_{\mathcal{X} \times \Theta} p \log\left(p \right) {\rm d}\P(\x,\theta).
\end{align*}  
In addition, for any $(\X, \theta)$, it also holds that $\left|q\right| \leq 1$. The lemma follows by combining the dominated convergence theorem with the continuous mapping theorem for the logarithm.
\end{proof}

\section{Comparison with Monte Carlo Synthetic Likelihood-Based Methods}\label{supp_mat: comp_gp_mc}



In this section we compare the performance of \texttt{ACORE} with Monte-Carlo (MC) synthetic likelihood-based methods, more specifically Gaussian process (GP) interpolation \cite{Frate2017GP_LFI}. 
The latter method first simulates multiple sample points for 
a few different values of $\theta$. For each fixed $\theta$, one fits a Gaussian synthetic likelihood function. The GP likelihood model is then used to smoothly interpolate across the parameter space by fitting a mean function $m(\theta)$ and a covariance function $\Sigma(\theta)$. As a note, \citet{Cranmer2019frontier_sim_based} point out that such MC methods are less efficient than methods that estimate the likelihood ratio directly because of the need to first estimate the entire likelihood.

For our comparison, we use the two toy examples described in Section~\ref{sec: toy_example} and Table~\ref{tab:toy_example_setup}. To allocate $B$ sample points for the GP interpolation, we use the following strategy: For $q \in \{5, 10 ,25\}$, first choose $\theta_1, ..., \theta_q$ on an evenly spaced grid across the parameter space. Then, generate $N=B/q$ sample points $\mathbf{X}_1, ... ,\mathbf{X}_N$ at each location $\theta$. 

Table~\ref{tab:comparison_with_MC} summarizes the results. Unlike Table~\ref{tab:cross_ent_poisson}, we do not report the cross-entropy loss because GP interpolation is not a classification algorithm; instead we report the mean squared error in estimating the likelihood ratio across the parameter space. Our results show that when the simulated data at each $\theta$ are approximately Gaussian, as in the Poisson example, MC-based GP interpolation provides a better approximation of the likelihood ratio due to its  parametric assumptions. However, when the parametric assumptions are not valid, as in the GMM example, MC-based GP fails to approximate the likelihood ratio regardless of how large $N$ or $B$ are. In such settings, we do better with a fully nonparametric approach. As a note, MC-based GP uses the asymptotic $\chi^2$ approximation by Wilks' theorem to determine the critical values of the confidence sets. In our experiments, using quantile regression for critical values instead (as in \texttt{ACORE}) led to a significant increase in power for the GP likelihood models:
from $\approx 0.48$ to $\approx 0.51$ for the Poisson example, and from $\approx 0.02$ to $\approx 0.2$ for the GMM example.

\begin{table*}[!ht]
\resizebox{0.49\textwidth}{!}{%
\begin{tabular}{|c|c|c|c|c|}
\multicolumn{5}{c}{Poisson Example} \\ \hline
\textit{$B$} & \textit{Classifier} & \textit{90 $\%$ Mean Squared} & \textit{Average} & \textit{Size of} \\
 & & \textit{Error Interval} & \textit{Power} & \textit{Confidence Set [\%]} \\ \hline
\multirow{6}{*}{100} & MLP & [2.14, 989.78] & 0.27 & 72.8 $\pm$ 16.4\\
 & NN & [4.14, 4074.65] & 0.25 & 75.6 $\pm$ 23.2\\
 & QDA & [0.41, 34.79] & 0.41 & \textbf{60.1 $\pm$ 14.9}\\
 & G.P. (5) & \textbf{[0.05, 4.09]} & 0.47 & \textbf{53.5 $\pm$ 9.2}\\
 & G.P. (10) & \textbf{[0.06, 4.97]} & \textbf{0.48} & \textbf{53.2 $\pm$ 10.7}\\
 & G.P. (25) & \textbf{[0.03, 6.54]} & \textbf{0.48} & \textbf{53.2 $\pm$ 10.8}\\ \hline
\multirow{6}{*}{500} & MLP & [0.86, 22.45] & 0.38 & 62.2 $\pm$ 19.1\\
 & NN & [1.95, 32.78] & 0.37 & \textbf{64.2 $\pm$ 17.3}\\
 & QDA & [0.08, 6.95] & 0.45 & \textbf{55.5 $\pm$ 10.8}\\
 & G.P. (5) & \textbf{[0.01, 0.81]} & \textbf{0.49} & \textbf{52.4 $\pm$ 5.6}\\
 & G.P. (10) & \textbf{[0.02, 0.85]} & \textbf{0.49} & \textbf{52.0 $\pm$ 5.4}\\
 & G.P. (25) & \textbf{[0.01, 1.12]} & \textbf{0.48} & \textbf{52.5 $\pm$ 6.0}\\ \hline
\multirow{6}{*}{1,000} & MLP & [0.81, 21.44] & 0.42 & \textbf{58.8 $\pm$ 17.0}\\
 & NN & [1.77, 17.88] & 0.45 & \textbf{56.1 $\pm$ 16.2}\\
 & QDA & [0.06, 2.83] & \textbf{0.49} & \textbf{52.1 $\pm$ 9.0}\\
 & G.P. (5) & \textbf{[0.01, 0.48]} & \textbf{0.49} & \textbf{52.3 $\pm$ 5.0}\\
 & G.P. (10) & \textbf{[0.01, 0.46]} & \textbf{0.48} & \textbf{52.5 $\pm$ 5.3}\\
 & G.P. (25) & \textbf{[0.01, 0.45]} & \textbf{0.48} & \textbf{52.6 $\pm$ 5.5}\\ \hline
 \hline
 - & {\color{red} \textbf{Exact}} & - &  {\color{red} \textbf{0.54}}  & {\color{red} \textbf{45.0 $\pm$ 4.9}} \\ \hline
\end{tabular} \\ [5ex]
}
\resizebox{0.5\textwidth}{!}{%
\begin{tabular}{|c|c|c|c|c|}
\multicolumn{5}{c}{GMM Example} \\ \hline
\textit{$B$} & \textit{Classifier} & \textit{90 $\%$ Mean Squared} & \textit{Average} & \textit{Size of}  \\
 & & \textit{Error Interval ($\times 10^3$)} & \textit{Power} & \textit{Confidence Set [\%]} \\ \hline
\multirow{6}{*}{100} & MLP & \textbf{[0.34, 1.46]} & \textbf{0.87} & \textbf{14.5 $\pm$ 4.5}\\
 & NN & [1.33, 11.77] & 0.49 & 52.1 $\pm$ 24.7\\
 & QDA & [2.88, 3.56] & 0.16 & 84.0 $\pm$ 21.8\\
 & G.P. (5) & [3.35, 3.82] & 0.02 & 97.7 $\pm$ 8.8\\
 & G.P. (10) & [3.34, 3.82] & 0.03 & 96.9 $\pm$ 9.5\\
 & G.P. (25) & [3.36, 3.82] & 0.02 & 98.2 $\pm$ 6.1\\ \hline
\multirow{6}{*}{500} & MLP & \textbf{[0.44, 1.35]} & \textbf{0.90} & 12.1 $\pm$ 2.8\\
 & NN & [0.99, 2.65] & 0.57 & 44.0 $\pm$ 23.3\\
 & QDA & [3.14, 3.73] & 0.16 & 83.8 $\pm$ 22.2\\
 & G.P. (5) & [3.39, 3.83] & 0.00 & 100.0 $\pm$ 0.0\\
 & G.P. (10) & [3.39, 3.83] & 0.01 & 99.1 $\pm$ 5.5\\
 & G.P. (25) & [3.38, 3.83] & 0.00 & 99.8 $\pm$ 1.5\\ \hline
\multirow{6}{*}{1,000} & MLP & \textbf{[0.53, 1.17]} & \textbf{0.90} & \textbf{12.1 $\pm$ 2.8}\\
 & NN & [0.57, 2.04] & 0.71 & 30.2 $\pm$ 18.5\\
 & QDA & [3.26, 3.94] & 0.14 & 85.7 $\pm$ 20.1\\
 & G.P. (5) & [3.39, 3.98] & 0.00 & 100.0 $\pm$ 0.0\\
 & G.P. (10) & [3.39, 3.98] & 0.00 & 100.0 $\pm$ 0.0\\
 & G.P. (25) & [3.39, 3.98] & 0.00 & 99.9 $\pm$ 1.2\\ \hline
 \hline
 - & {\color{red} \textbf{Exact}} & - &  {\color{red} \textbf{0.92}}  & {\color{red} \textbf{9.5 $\pm$ 2.0}} \\ \hline
\end{tabular}
}
\caption{
Results for \texttt{ACORE} (MLP, NN, QDA) and Gaussian Process  interpolation (GP for $q=5, 10, 25$; see text) for the two toy examples, Poisson example (left) and GMM example (right), of Section~\ref{sec: toy_example}.  The tables list the mean squared error (MSE) between the estimated and true likelihood, the power (averaged over $\theta$) and the size of confidence sets, for different values of $B$ and for different classifiers. We report a 90\% confidence interval for the MSE, together with the mean and standard deviation of the size of the estimated 90\% confidence set for $\theta$. Best results for each training sample size $B$ are marked in bold-faced. 
All fitted classifiers produce valid $90\%$ confidence sets for $\theta$ according to our diagnostics.
}
\label{tab:comparison_with_MC}
\end{table*}

\section{Comparison with Calibrated Approximate Ratio of Likelihood Classifiers}\label{supp_mat: comp_carl}


In this section we compare the performance of \texttt{ACORE} with the calibrated approximate ratio of likelihood (CARL) estimator by \citet{Cranmer2015LikRatio}. CARL approximates the likelihood ratio $\Lambda(D; \Theta_0) = \mathcal{L}(D; \theta_0)/\mathcal{L}(D; \theta_1)$ by turning the density ratio  estimation into a supervised classification problem,   where a probabilistic classifier is trained to separate samples from $F_{\theta_0}$ and $F_{\theta_1}$. As such, CARL classifiers are ``doubly parameterized'' by $\theta_0$ and $\theta_1$, whereas the \texttt{ACORE} classifier is parameterized by a {\em single} parameter $\theta$ in the definition of the odds of $F_\theta$ versus $G$. 

In our study, we include three different CARL classifiers, implemented with the \texttt{MADMINER} neural network-based software \citep{Brehmer2019madminer, Brehmer2019MadMiner_code}: (a) a shallow perceptron with 100 neurons (equivalent to the MLP used in Section~\ref{sec: toy_example}), (b) a 2-layer deep network with 20 neurons per layer, and (c) a 2-layer deep network with 20 and 50 neurons in the two layers respectively.\footnote{Changing the number of neurons per layers did not seem to provide a significant difference in performance for the 2-layer deep networks. Number of epochs and learning rate were manually tuned (with a search in the range $[20, 200]$ and $10^{\{ -6, ..., -2\}}$ respectively).} To allocate $B$ sample points for interpolation we devised two schemes: (i) a uniform sampling, and (ii) a Monte Carlo sampling over the parameter space.   For (i), we uniformly sample $B$ parameters and then generate a sample point $\mathbf{X}$ at each parameter value. For (ii), we first select evenly  spaced parameters $\theta_{0,1}, ..., \theta_{0,q}$ and $\theta_{1,1} ..., \theta_{1, q}$, for the numerator and the denominator respectively. We set  $q \in \{10, 20, 30\}$, resulting in $N=B/q$  sample points $\mathbf{X}_1, ..., \mathbf{X}_N$  at each $\theta$ location. Because the $\chi^2$ approximation by Wilks' theorem did not yield valid confidence sets for CARL classifiers, we computed critical values as in \texttt{ACORE} Algorithm~\ref{alg:estimate_thresholds2}.

Table~\ref{tab:comparison_with_carl} shows the results of \texttt{ACORE} and \texttt{CARL} for the synthetic data in Section~\ref{sec: toy_example} and Table~\ref{tab:toy_example_setup}. For both the Poisson and GMM examples, CARL classifiers yield a higher mean squared error in estimating the likelihood ratio, as well as lower power and larger confidence intervals.

\begin{table*}[!ht]
\resizebox{.495\textwidth}{!}{%
\begin{tabular}{|c|c|c|c|c|}
\multicolumn{5}{c}{Poisson Example} \\ \hline
\textit{$B$} & \textit{Classifier} & \textit{90 $\%$ Mean Squared} & \textit{Average} & \textit{Size of} \\
 & & \textit{Error Interval} & \textit{Power} & \textit{Confidence Set [\%]} \\ \hline
\multirow{9}{*}{200} & MLP & [3.25, 1305.45] & 0.17 & 82.7 $\pm$ 15.0\\
 & NN & [2.88, 185.47] & 0.34 & 66.9 $\pm$ 20.7\\
 & QDA & \textbf{[0.20, 25.16]} & \textbf{0.45} & \textbf{55.8 $\pm$ 13.2}\\ \cline{2-5}
 & MLP (MC) & [2.51, 38.10] & 0.24 & 76.1 $\pm$ 21.3\\
 & (20,20) DNN (MC) & [2.53, 25.41] & 0.19 & 80.9 $\pm$ 17.8\\
 & (50,20) DNN (MC) & [2.76, 26.00] & 0.19 & 81.3 $\pm$ 17.8\\ \cline{2-5}
 & MLP (U) & [2.03, 45.19] & 0.19 & 81.3 $\pm$ 19.2\\
 & (20,20) DNN (U) & [2.95, 19.76] & 0.24 & 76.6 $\pm$ 19.8\\
 & (50,20) DNN (U) & [2.43, 18.72] & 0.23 & 77.8 $\pm$ 20.1\\ \hline \hline
\multirow{9}{*}{800} & MLP & [1.69, 450.81] & 0.27 & 73.0 $\pm$ 20.1\\
 & NN & [1.47, 19.32] & 0.42 & \textbf{59.2 $\pm$ 15.9}\\
 & QDA & \textbf{[0.04, 5.03]} & \textbf{0.49} & \textbf{52.0 $\pm$ 9.3}\\ \cline{2-5}
 & MLP (MC) & [2.38, 24.50] & 0.22 & 78.5 $\pm$ 21.0\\
 & (20,20) DNN (MC) & [2.49, 21.49] & 0.25 & 75.3 $\pm$ 18.8\\
 & (50,20) DNN (MC) & [2.52, 18.13] & 0.23 & 76.9 $\pm$ 20.1\\ \cline{2-5}
 & MLP (U) & [2.04, 23.24] & 0.20 & 79.9 $\pm$ 17.4\\
 & (20,20) DNN (U) & [2.48, 17.36] & 0.22 & 77.9 $\pm$ 17.6\\
 & (50,20) DNN (U) & [2.25, 17.87] & 0.21 & 78.9 $\pm$ 20.0\\ \hline \hline
\multirow{9}{*}{1,800} & MLP & [0.81, 19.11] & 0.37 & \textbf{63.7 $\pm$ 21.1}\\
 & NN & [1.09, 11.27] & 0.44 & \textbf{56.9 $\pm$ 14.3}\\
 & QDA & \textbf{[0.03, 1.60]} & \textbf{0.50} & \textbf{51.0 $\pm$ 6.6}\\ \cline{2-5}
 & MLP (MC) & [2.13, 35.39] & 0.18 & 82.4 $\pm$ 17.7\\
 & (20,20) DNN (MC) & [2.74, 28.15] & 0.20 & 80.3 $\pm$ 19.7\\
 & (50,20) DNN (MC) & [2.62, 28.15] & 0.18 & 81.9 $\pm$ 19.5\\ \cline{2-5}
 & MLP (U) & [2.15, 25.51] & 0.19 & 81.4 $\pm$ 19.8\\
 & (20,20) DNN (U) & [2.34, 15.93] & 0.23 & 77.0 $\pm$ 22.6\\
 & (50,20) DNN (U) & [2.38, 17.97] & 0.19 & 81.6 $\pm$ 17.2\\ \hline
\hline
- & {\color{red} \textbf{Exact}} & - &  {\color{red} \textbf{0.54}}  & {\color{red} \textbf{45.0 $\pm$ 4.9}} \\ \hline
\end{tabular} \\ [5ex]
}
\resizebox{.505\textwidth}{!}{%
\begin{tabular}{|c|c|c|c|c|}
\multicolumn{5}{c}{GMM Example} \\ \hline
\textit{$B$} & \textit{Classifier} & \textit{90 $\%$ Mean Squared} & \textit{Average} & \textit{Size of}  \\
 & & \textit{Error Interval ($\times 10^{3}$)} & \textit{Power} & \textit{Confidence Set [\%]} \\ \hline
\multirow{9}{*}{200} & MLP & \textbf{[0.56, 1.69]} & \textbf{0.88} & \textbf{14.2 $\pm$ 8.2}\\
 & NN & [1.13, 4.17] & 0.50 & 51.5 $\pm$ 24.8\\
 & QDA & [3.05, 3.63] & 0.12 & 87.6 $\pm$ 19.7\\ \cline{2-5}
 & MLP (MC) & [3.03, 3.61] & 0.27 & 73.5 $\pm$ 20.5\\
 & (20,20) DNN (MC) & [3.13, 3.70] & 0.25 & 75.6 $\pm$ 20.0\\
 & (50,20) DNN (MC) & [3.16, 3.67] & 0.28 & 72.8 $\pm$ 19.6\\ \cline{2-5}
 & MLP (U) & [3.01, 3.72] & 0.30 & 70.2 $\pm$ 21.2\\
 & (20,20) DNN (U) & [3.18, 3.87] & 0.24 & 76.3 $\pm$ 21.5\\
 & (50,20) DNN (U) & [3.12, 3.92] & 0.27 & 73.0 $\pm$ 21.2\\ \hline \hline
\multirow{9}{*}{800} & MLP & \textbf{[0.89, 1.59]} & \textbf{0.90} & \textbf{12.1 $\pm$ 2.5}\\
 & NN & [0.78, 2.31] & 0.69 & 32.0 $\pm$ 18.9\\
 & QDA & [3.23, 3.66] & 0.14 & 86.1 $\pm$ 20.4\\ \cline{2-5}
 & MLP (MC) & [3.02, 3.58] & 0.30 & 70.8 $\pm$ 20.4\\
 & (20,20) DNN (MC) & [3.10, 3.63] & 0.27 & 73.6 $\pm$ 20.2\\
 & (50,20) DNN (MC) & [3.03, 3.47] & 0.30 & 70.5 $\pm$ 18.5\\ \cline{2-5}
 & MLP (U) & [3.01, 3.62] & 0.26 & 74.7 $\pm$ 20.6\\
 & (20,20) DNN (U) & [3.12, 3.64] & 0.26 & 74.4 $\pm$ 19.2\\
 & (50,20) DNN (U) & [3.00, 3.56] & 0.29 & 71.8 $\pm$ 19.9\\ \hline \hline
\multirow{9}{*}{1,800} & MLP & \textbf{[0.33, 1.55]} & \textbf{0.90} & \textbf{11.5 $\pm$ 2.6}\\
 & NN & \textbf{[0.32, 1.57]} & 0.83 & \textbf{19.3 $\pm$ 10.3}\\
 & QDA & [3.29, 3.81] & 0.16 & 83.7 $\pm$ 22.2\\ \cline{2-5}
 & MLP (MC) & [2.99, 3.54] & 0.33 & 67.5 $\pm$ 19.6\\
 & (20,20) DNN (MC) & [3.02, 3.54] & 0.31 & 69.7 $\pm$ 19.3\\
 & (50,20) DNN (MC) & [2.95, 3.51] & 0.38 & 63.1 $\pm$ 15.9\\ \cline{2-5}
 & MLP (U) & [2.99, 3.45] & 0.33 & 67.7 $\pm$ 17.0\\
 & (20,20) DNN (U) & [3.02, 3.56] & 0.33 & 67.3 $\pm$ 18.0\\
 & (50,20) DNN (U) & [2.98, 3.41] & 0.38 & 63.1 $\pm$ 15.3\\ \hline
 \hline
 - & {\color{red} \textbf{Exact}} & - &  {\color{red} \textbf{0.92}}  & {\color{red} \textbf{9.5 $\pm$ 2.0}} \\ \hline
\end{tabular}
}
\caption{
Results for \texttt{ACORE} (MLP, NN, QDA) and 
CARL or uniform (U) and Monte-Carlo (MC) sampling schemes 
in the Poisson example (left) and GMM example (right) 
settings  of Section~\ref{sec: toy_example}.  The tables list the mean squared error (MSE) between the estimated and true likelihood, the power (averaged over $\theta$) and the size of confidence sets, for different values of $B$ and for different classifiers. We report a 90\% confidence interval for the MSE, together with the mean and standard deviation of the size of the estimated 90\% confidence set for $\theta$. The best results for each training sample size $B$ are marked in bold-faced.
}
\label{tab:comparison_with_carl}
\end{table*}

\section{Runtime Analysis}\label{supp_mat: computational_costs}

In this section we provide a runtime analysis for constructing one \texttt{ACORE} confidence set for the two examples  in Section~\ref{sec: toy_example} and Table~\ref{tab:cross_ent_poisson}. We also provide a running time comparison with the two methods described in Sections~\ref{supp_mat: comp_gp_mc} and \ref{supp_mat: comp_carl}. This analysis was performed on  a 8-Core Intel Xeon 3.33GHz X5680 CPU.

The procedure for constructing confidence sets with \texttt{ACORE} is  outlined in Algorithm~\ref{alg:conf_reg}. In this analysis we break  the computation into 4 steps: (i) odds ratio training as described by Algorithm~\ref{alg:joint_y}, (ii) computing the test statistic \eqref{eq:odds_ratio_statistic} for the observed data, (iii) computing the test statistic \eqref{eq:odds_ratio_statistic} in the $B^{'}$ sample as described by Algorithm~\ref{alg:estimate_thresholds_conf2} and (iv) quantile regression algorithm training. Table~\ref{tab:running_times} summarizes our running times results. \texttt{ACORE} constructs one confidence set in less than $20$ and $30$ seconds for Poisson and GMM examples respectively. The main computational bottleneck is step (iii), while the computation time of step (i) increases with the sample size $B$.

Figure~\ref{fig:running_times_comp} shows the results of comparing confidence set  construction runtimes with MC GP and CARL classifiers. For both the Poisson and the GMM examples, we only consider the best performing \texttt{ACORE} classifiers, and the two CARL classifiers with 20 hidden units in both layers. Results show  \texttt{ACORE} classifiers are comparable with GP interpolation in terms of running times, while CARL classifiers tend to have significantly longer runtimes.

\begin{table*}[!ht]
\resizebox{0.98\textwidth}{!}{%
\begin{tabular}{|c|c|c|c|c|c|c|}
\hline
\multicolumn{7}{c}{Running Times to Generate a Confidence Set (Seconds) -- Poisson Example} \\ \hline
\textit{$B$} & \textit{Classifier} & \textit{Odds Ratio} & \textit{Odds Ratio} & \textit{Calculate \eqref{eq:odds_ratio_statistic} for} & \textit{Quantile Regression} & \textit{Total Running}\\ 
& & \textit{Training} & \textit{Prediction} & \textit{$B^{'}$ Samples} & \textit{Training} & \textit{Time} \\ \hline
\multirow{3}{*}{100} & MLP & 0.38 $\pm$ 0.31 & 0.42 $\pm$ 0.10 & 10.40 $\pm$ 0.71 & 0.66 $\pm$ 0.28 & 11.86 $\pm$ 1.02\\
 & NN & 0.03 $\pm$ 0.01 & 0.35 $\pm$ 0.12 & 9.83 $\pm$ 4.99 & 0.82 $\pm$ 0.67 & 11.02 $\pm$ 5.73\\
 & QDA & 0.02 $\pm$ 0.01 & 0.18 $\pm$ 0.11 & 4.50 $\pm$ 2.65 & 0.58 $\pm$ 0.21 & 5.29 $\pm$ 2.96\\ \hline
\multirow{3}{*}{500} & MLP & 1.62 $\pm$ 0.39 & 0.46 $\pm$ 0.04 & 11.49 $\pm$ 0.45 & 0.68 $\pm$ 0.09 & 14.26 $\pm$ 0.61\\
 & NN & 0.13 $\pm$ 0.01 & 0.54 $\pm$ 0.03 & 13.28 $\pm$ 0.26 & 0.66 $\pm$ 0.04 & 14.60 $\pm$ 0.29\\
 & QDA & 0.13 $\pm$ 0.01 & 0.16 $\pm$ 0.01 & 4.12 $\pm$ 0.09 & 0.65 $\pm$ 0.06 & 5.05 $\pm$ 0.14\\ \hline
\multirow{3}{*}{1,000} & MLP & 2.65 $\pm$ 0.88 & 0.48 $\pm$ 0.08 & 11.93 $\pm$ 1.93 & 0.73 $\pm$ 0.06 & 15.79 $\pm$ 2.30\\
 & NN & 0.24 $\pm$ 0.04 & 0.77 $\pm$ 0.21 & 17.90 $\pm$ 2.82 & 0.67 $\pm$ 0.10 & 19.59 $\pm$ 2.83\\
 & QDA & 0.27 $\pm$ 0.08 & 0.17 $\pm$ 0.05 & 4.37 $\pm$ 1.02 & 0.64 $\pm$ 0.16 & 5.45 $\pm$ 1.29\\
\hline
\end{tabular} \\ [5ex]
}\newline
\vspace*{.1cm}
\newline
\resizebox{0.98\textwidth}{!}{%
\begin{tabular}{|c|c|c|c|c|c|c|}
\hline
\multicolumn{7}{c}{Running Times to Generate a Confidence Set (Seconds) -- GMM Example} \\ \hline
\textit{$B$} & \textit{Classifier} & \textit{Odds Ratio} & \textit{Odds Ratio} & \textit{Calculate \eqref{eq:odds_ratio_statistic} for} & \textit{Quantile Regression} & \textit{Total Running}\\ 
& & \textit{Training} & \textit{Prediction} & \textit{$B^{'}$ Samples} & \textit{Training} & \textit{Time} \\ \hline
\multirow{3}{*}{100} & MLP & 5.89 $\pm$ 1.66 & 0.45 $\pm$ 0.18 & 10.79 $\pm$ 2.06 & 0.60 $\pm$ 0.21 & 17.74 $\pm$ 3.92\\
 & NN & 0.03 $\pm$ 0.00 & 0.29 $\pm$ 0.06 & 8.60 $\pm$ 2.84 & 0.61 $\pm$ 0.18 & 9.53 $\pm$ 3.05\\
 & QDA & 0.03 $\pm$ 0.01 & 0.14 $\pm$ 0.04 & 3.81 $\pm$ 1.38 & 0.52 $\pm$ 0.14 & 4.50 $\pm$ 1.57\\ \hline
\multirow{3}{*}{500} & MLP & 9.89 $\pm$ 1.34 & 0.43 $\pm$ 0.06 & 11.64 $\pm$ 0.64 & 0.69 $\pm$ 0.06 & 22.64 $\pm$ 1.83\\
 & NN & 0.17 $\pm$ 0.01 & 0.52 $\pm$ 0.04 & 13.11 $\pm$ 0.79 & 0.63 $\pm$ 0.07 & 14.43 $\pm$ 0.85\\ 
 & QDA & 0.16 $\pm$ 0.01 & 0.15 $\pm$ 0.02 & 4.05 $\pm$ 0.26 & 0.59 $\pm$ 0.08 & 4.94 $\pm$ 0.35\\ \hline
\multirow{3}{*}{1,000} & MLP & 13.40 $\pm$ 2.60 & 0.47 $\pm$ 0.09 & 11.76 $\pm$ 0.79 & 0.68 $\pm$ 0.11 & 26.31 $\pm$ 3.36\\
 & NN & 0.34 $\pm$ 0.09 & 0.70 $\pm$ 0.11 & 17.15 $\pm$ 1.90 & 0.71 $\pm$ 0.17 & 18.90 $\pm$ 2.06\\
 & QDA & 0.32 $\pm$ 0.05 & 0.17 $\pm$ 0.05 & 4.75 $\pm$ 1.26 & 0.62 $\pm$ 0.07 & 5.87 $\pm$ 1.36\\ \hline
\end{tabular}
}
\caption{
Runtimes in seconds for constructing a confidence set with \texttt{ACORE}  
for the 
Poisson example (top) and GMM example (bottom). The 
 procedure for constructing confidence sets is 
 outlined  in Algorithm~\ref{alg:conf_reg}, and is split in 4 steps (see text). The rightmost column  shows total runtimes.
}
\label{tab:running_times}
\end{table*}

\begin{figure*}[!ht]
    \centering
    \includegraphics[width=0.775\textwidth]{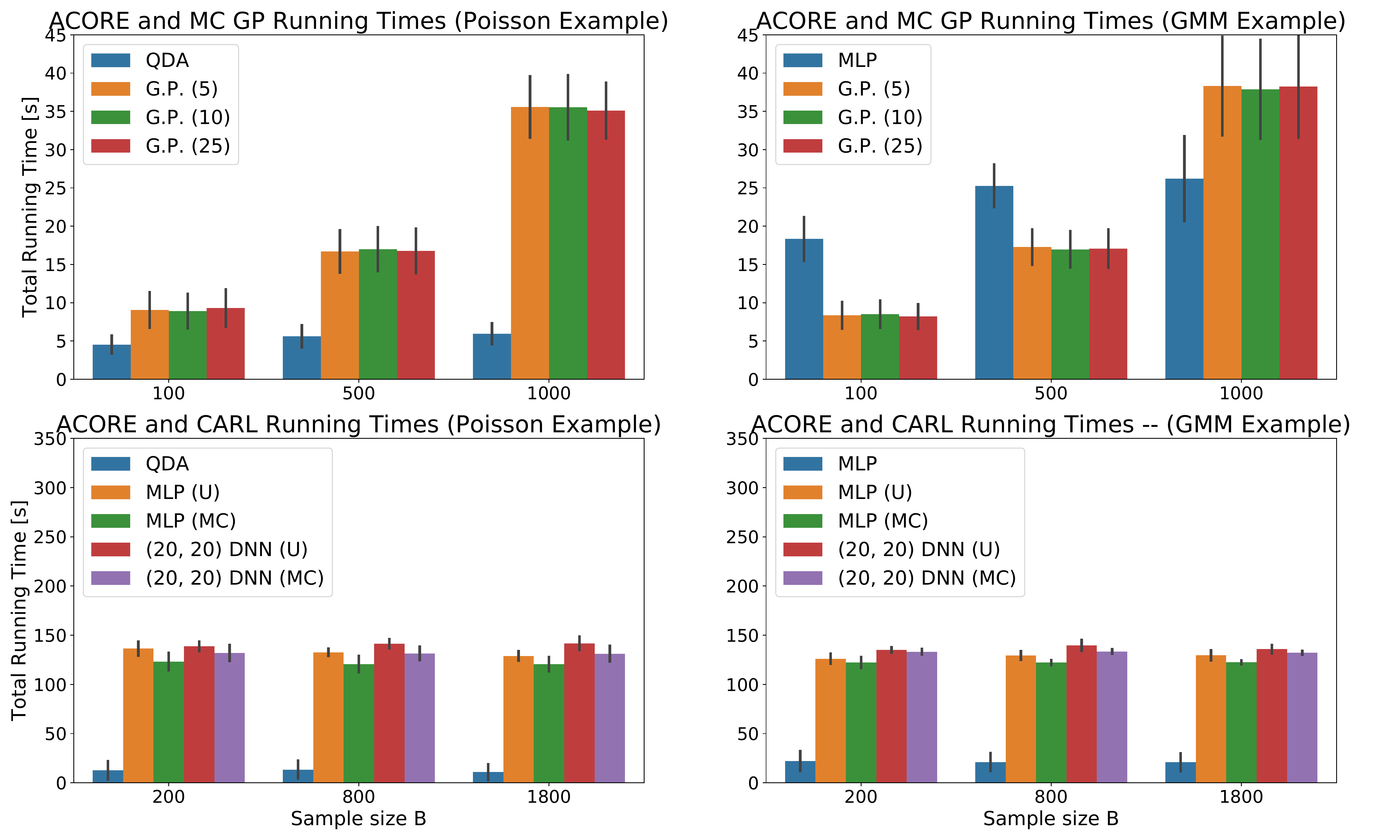}
    \caption{
    Runtimes in seconds for constructing a confidence set  for the
    Poisson example (left panels) and GMM example (right panels). The best \texttt{ACORE} classifier runtime is compared with Gaussian process interpolation (GP) for $q=\{5,10,25\}$, and the two smaller CARL classifiers
    for both sampling schemes. See text for details. Confidence bars are built with a one standard deviation interval around the mean.
    }
    \label{fig:running_times_comp}
\end{figure*}

\end{document}